\PassOptionsToPackage{unicode}{hyperref}
\PassOptionsToPackage{hyphens}{url}
\PassOptionsToPackage{dvipsnames,svgnames,x11names}{xcolor}
\documentclass[
  12pt]{article}

  \usepackage{amsmath,amssymb}
  \usepackage{mathtools}
  \DeclareMathOperator*{\argmin}{arg\,min}
  
\usepackage{iftex}
\ifPDFTeX
  \usepackage[T1]{fontenc}
  \usepackage[utf8]{inputenc}
  \usepackage{textcomp} 
\else 
  \usepackage{unicode-math}
  \defaultfontfeatures{Scale=MatchLowercase}
  \defaultfontfeatures[\rmfamily]{Ligatures=TeX,Scale=1}
\fi
\usepackage{lmodern}
\ifPDFTeX\else  
\fi
\IfFileExists{upquote.sty}{\usepackage{upquote}}{}
\IfFileExists{microtype.sty}{
  \usepackage[]{microtype}
  \UseMicrotypeSet[protrusion]{basicmath} 
}{}
\makeatletter
\@ifundefined{KOMAClassName}{
  \IfFileExists{parskip.sty}{%
    \usepackage{parskip}
  }{
    \setlength{\parindent}{0pt}
    \setlength{\parskip}{6pt plus 2pt minus 1pt}}
}{
  \KOMAoptions{parskip=half}}
\makeatother
\usepackage{xcolor}
\makeatletter
\ifx\paragraph\undefined\else
  \let\oldparagraph\paragraph
  \renewcommand{\paragraph}{
    \@ifstar
      \xxxParagraphStar
      \xxxParagraphNoStar
  }
  \newcommand{\xxxParagraphStar}[1]{\oldparagraph*{#1}\mbox{}}
  \newcommand{\xxxParagraphNoStar}[1]{\oldparagraph{#1}\mbox{}}
\fi
\ifx\subparagraph\undefined\else
  \let\oldsubparagraph\subparagraph
  \renewcommand{\subparagraph}{
    \@ifstar
      \xxxSubParagraphStar
      \xxxSubParagraphNoStar
  }
  \newcommand{\xxxSubParagraphStar}[1]{\oldsubparagraph*{#1}\mbox{}}
  \newcommand{\xxxSubParagraphNoStar}[1]{\oldsubparagraph{#1}\mbox{}}
\fi
\makeatother

\usepackage{longtable,booktabs,array}
\usepackage{calc} 
\usepackage{etoolbox}
\makeatletter
\patchcmd\longtable{\par}{\if@noskipsec\mbox{}\fi\par}{}{}
\makeatother
\IfFileExists{footnotehyper.sty}{\usepackage{footnotehyper}}{\usepackage{footnote}}
\makesavenoteenv{longtable}
\usepackage{graphicx}
\makeatletter
\def\maxwidth{\ifdim\Gin@nat@width>\linewidth\linewidth\else\Gin@nat@width\fi}
\def\maxheight{\ifdim\Gin@nat@height>\textheight\textheight\else\Gin@nat@height\fi}
\makeatother
\setkeys{Gin}{width=\maxwidth,height=\maxheight,keepaspectratio}
\makeatletter
\def\fps@figure{htbp}
\makeatother

\makeatletter
\@ifpackageloaded{caption}{}{\usepackage{caption}}
\AtBeginDocument{%
\ifdefined\contentsname
  \renewcommand*\contentsname{Table of contents}
\else
  \newcommand\contentsname{Table of contents}
\fi
\ifdefined\listfigurename
  \renewcommand*\listfigurename{List of Figures}
\else
  \newcommand\listfigurename{List of Figures}
\fi
\ifdefined\listtablename
  \renewcommand*\listtablename{List of Tables}
\else
  \newcommand\listtablename{List of Tables}
\fi
\ifdefined\figurename
  \renewcommand*\figurename{Figure}
\else
  \newcommand\figurename{Figure}
\fi
\ifdefined\tablename
  \renewcommand*\tablename{Table}
\else
  \newcommand\tablename{Table}
\fi
}
\@ifpackageloaded{float}{}{\usepackage{float}}
\floatstyle{ruled}
\@ifundefined{c@chapter}{\newfloat{codelisting}{h}{lop}}{\newfloat{codelisting}{h}{lop}[chapter]}
\floatname{codelisting}{Listing}

\makeatother
\makeatletter
\@ifpackageloaded{caption}{}{\usepackage{caption}}
\@ifpackageloaded{subcaption}{}{\usepackage{subcaption}}
\makeatother

\ifLuaTeX
  \usepackage{selnolig}  
\fi
\usepackage[]{natbib}
\usepackage{bookmark}
\usepackage{amsmath}
\usepackage{amsthm}
\usepackage{graphicx}
\usepackage{amssymb}
\usepackage{comment}

\IfFileExists{xurl.sty}{\usepackage{xurl}}{} 
\hypersetup{
  pdftitle={Title},
  pdfauthor={Author 1; Author 2},
  pdfkeywords={3 to 6 keywords, that do not appear in the title},
  colorlinks=true,
  linkcolor={blue},
  filecolor={Maroon},
  citecolor={Blue},
  urlcolor={Blue},
  pdfcreator={LaTeX via pandoc}}

\newcommand{\anon}{1}

\newcommand{\ndata}{{\cal D}^{(n)}}
\newcommand{\mdata}{{\cal D}^{(m)}}

\newcommand{\ndatai}{{\cal D}^{(n) \backslash i}}
\newcommand{\mdatai}{{\cal D}^{(m) \backslash i}}

\newcommand{\iid}{\stackrel{\text{i.i.d.}}{\sim}}
\newcommand{\ac}{\mathcal{F}}
\newcommand{\mc}{\mathcal{H}}
\newcommand{\E}{\mathbb{E}}
\newtheorem{corollary}{Corollary}

\usepackage{tabularx,ragged2e}
\newcolumntype{L}{>{\RaggedRight\arraybackslash}X} 
\usepackage{enumitem}
\newlist{tabitemize}{itemize}{1}
\setlist[tabitemize,1]{
    label=\textbullet, 
    left=0pt, 
    nosep,
    before={\begin{minipage}[t]{\linewidth}},
    after ={\end{minipage}}
}

\newtheorem{proposition}{Proposition}
\newtheorem{theorem}{Theorem}

\begin{document}

\graphicspath{{plots}}

\newif\ifhideauthors
\hideauthorsfalse    

\newcommand{\hideable}[1]{%
  \ifhideauthors
    \textcolor{gray}{[redacted author information]}%
  \else
    #1%
  \fi
}

\def\spacingset#1{\renewcommand{\baselinestretch}%
{#1}\small\normalsize} \spacingset{1}


\if1\anon
{
  \title{\bf Model Class Selection}
  \author{\hideable{Ryan Cecil and Lucas Mentch \hspace{.2cm}\\
    Department of Statistics, University of Pittsburgh}\\}
  \maketitle
} \fi

\if0\anon
{
  \bigskip
  \bigskip
  \bigskip
  \begin{center}
    {\LARGE\bf Title}
\end{center}
  \medskip
} \fi

\bigskip
\begin{abstract}
Classical model selection seeks to find a single model within a particular class that optimizes some pre-specified criteria, such as maximizing a likelihood or minimizing a risk. More recently, there has been an increased interest in model set selection (MSS), where the aim is to identify a (confidence) set of near-optimal models. Here, we generalize the MSS framework further by introducing the idea of model class selection (MCS). In MCS, multiple model collections are evaluated, and all collections that contain at least one optimal model are sought for identification. Under mild conditions, data splitting based approaches are shown to provide general solutions for MCS. As a direct consequence, for particular datasets we are able to investigate formally whether classes of simpler and more interpretable statistical models are able to perform on par with more complex black-box machine learning models. A variety of simulated and real-data experiments are provided.

\end{abstract}

\noindent%
\vfill

\newpage
\spacingset{1.8} 

\section{Introduction}

In statistical and machine learning (ML) studies, it is standard to use model evaluation metrics such as AIC, BIC, step-wise selection, holdout, or cross-validation to identify a single empirically optimal model \citep{hastie_elements_2009}. These methods are used to compare the predictive performances of various modeling algorithms often associated with different model classes. For example, suppose a data scientist wishes to construct a model that can predict if a loan applicant will default in less than a year based on the available data at their company. Given the data, they know strategies for fitting various types of models such as an additive linear model, random forest, and neural network. They are likely unsure, however, of which type of model will perform best.  A common approach is thus to simply identify the model type that minimizes some empirical metric (e.g.\ cross-validation error) and re-fit on the full dataset to obtain a final model for inference or prediction.

Oftentimes, many well-performing models exist that optimize the chosen model selection criteria. This concept, called the Rashomon effect, may lead researchers with different datasets to construct entirely different models \citep{breiman_statistical_2001}.
In recognition of this effect, procedures for model \emph{set} selection (MSS) have been growing in popularity. Such MSS methods generally seek to identify an entire collection of similarly optimal models (i.e.\ an approximate \emph{Rashomon set}) rather than relying on only the single one identified via classical model selection methods (e.g. see \cite{fisher_all_2019,wasserman_universal_2020,kissel_forward_2024}).  Frequently, these sets of models are computed such that they are guaranteed to contain the optimal model(s) of a given type with a desired level of confidence under certain conditions and sample sizes. The size of the Rashomon set can be seen as providing a measure of the uncertainty of the model selection problem. In addition, the computed model set can be used to explore the data and identify important predictors \citep{fisher_all_2019}.

Recently, ML methods have achieved significant breakthroughs in a variety of fields, surpassing benchmark performances once thought to be unattainable. These advances have been driven in large part by the modern growth of computational resources, allowing researchers to train large scale models with potentially billions of parameters. Compute at scale, however, often comes with a cost. Unlike classical statistical models that are interpretable, have quantifiable uncertainty, and are able to be used to conduct valid inference, the theoretical underpinnings of ML approaches are not fully developed \citep{drazen_where_2023}. Recent work in mathematics, statistics, and computer science have made significant strides towards addressing these issues (e.g. \cite{abdar_review_2021,kutyniok_modern_2022,mentch_quantifying_2016,rudin_interpretable_2022}). However, there still exists a significant trade-off in interpretability and theoretical guarantees when one moves from simple, well-studied classes of models (e.g. linear models) to complex classes (e.g. random forests or neural networks) to improve predictive accuracy. This is often an issue in studies focused on inference or causality and settings where the model is employed in high-stakes environments such as healthcare or criminal justice \citep{wexler_opinion_2017, varshney_safety_2017,rudin_stop_2019}.

On the other hand, it is also somewhat common for practitioners to rule out classical models \emph{a priori} whenever the data is assumed to be too large or complex.  Recent research, however, suggests that this may be frequently misguided as there are numerous real-world settings in which simpler, interpretable models can perform on par with ML alternatives \citep{rudin_stop_2019,semenova_existence_2022,semenova_path_2023,boner_using_2024}. When this is the case, researchers may naturally prefer the simpler class in most settings. On the other hand, if complex ML models significantly outperform simpler alternatives, then there may be a preference for performance over interpretability. To our knowledge, little theoretical attention seems to have been given to formally detecting when these trade-offs exist. To address this, in this work, we introduce model class selection (MCS) methods that are able to determine with high probability when no model within a simple class can match the performance gains realized by more complex ML models. 

We formally investigate this problem in Section \ref{ch:mcs-construction} by proposing a general MCS hypothesis testing framework. In Section \ref{sec:proposed-methods}, we extend data splitting based tools to this general MCS setting. In Sections \ref{sec:alternative-approaches} and \ref{sec:experiments}, we consider alternative approaches to MCS and examine the power of our proposed methodology in simulated and real-world settings before concluding with a discussion in Section \ref{sec:conclusions}.



\section{Constructing Methods for MCS}
\label{ch:mcs-construction}
Denote the data as \( \ndata = \left \{  Z_1, \dots, Z_n \right \}  \) where \(Z_1, \dots, Z_n \iid P \) are independent and identically distributed with support over the set \( \mathcal{Z} \). Let  \( \mathcal{H} \subseteq \{h \;|\; h: \mathcal{Z} \rightarrow \mathcal{Y} \}  \) denote a model class, defined as a set of measurable functions mapping from the data space \( \mathcal{Z} \) to a prediction space \( \mathcal{Y} \). Define $\ell: \mc \times \mathcal{Z}  \rightarrow \mathbb{R}^{+}$ to be a measurable function that specifies the error of a model \( h  \in \mc\) on some data point \( z \in \mathcal{Z} \).  The goal of classical model selection is to choose a model \( h \in \mc\) that minimizes the risk \( R(h) = \E_{Z_{0}}[\ell(h, Z_{0})] \),  defined as the expected loss of \( h \) at an independent test point \( Z_0 \sim P \).

To formally establish the goals of MCS, for \( j = 1,...,d \) where \( d \ge 2 \), let \( T_j \) denote a non-empty index set corresponding to models \( h_t \in \mc \) for any \( t \in T_j \). Define \( \mc_j \coloneq \left \{ h_{t} : t \in T_{j} \right \} \) for \( j = 1,...,d \) to be the model classes of interest. Denoting \( [d] \coloneq \left \{ 1,...,d \right \} \), our goal is to form a confidence set for
\[ \Theta = \Theta \left ( P \right ) \coloneq \argmin_{j \in [d]} \inf_{t \in T_j} R \left ( h_t \right ) \]
which is the set of all indices corresponding to model classes containing at least one model that achieves the optimal risk. In other words, given \( r \in [d] \), if \( r \in \Theta \), then the model class \( \mc_r \) contains a model that performs as well or better than any other model from any class. On the other hand, if \( r \not \in \Theta \), then the model class \( \mc_r \) should not be considered if the goal is to choose a model with the best possible performance. Although the problem above may look complex in nature, we will show that the process of constructing valid confidence sets may be simplified by focusing on hypothesis tests for MSS.

\subsection{Hypothesis Tests}
For simplicity, we consider a dual hypothesis testing problem. Given \( r \in [d] \), we evaluate the hypotheses
\[ H_{0,r}: r \in \Theta \text{\;\;\;vs\;\;\;} H_{A,r}: r \not \in \Theta. \]

Let \( \psi_r: \left \{ Z_1, \dots, Z_{n} \right \} \rightarrow \left \{ 0,1 \right \}  \) denote a test function that rejects the null hypothesis when \( \psi_r = 1 \). Assume \( P \in \mathcal{P} \) where \( \mathcal{P} \) is a class of distributions. A test \( \psi_r \) with appropriate control of the type I error satisfies
\begin{equation} \left ( \overset{(as)}{\limsup_{n \to \infty}} \right )  \sup_{P \in \mathcal{P}_{0,r}}  P \left ( \psi_r =  1 \right )  \le \alpha \label{eq:test-condition}\end{equation}

where \( \mathcal{P}_{0,r} \subseteq \mathcal{P} \) denotes the distributions under which \( H_{0,r} \) is true. The paranthetical limit \( (as) \) in (\ref{eq:test-condition}) distinguishes two versions of the condition. A finite-sample condition would require that \( \sup_{P \in \mathcal{P}_{0,r}}  P \left ( \psi_r =  1 \right )  \le \alpha \), while an asymptotic version would require that \( \limsup_{n \to \infty}  \sup_{P \in \mathcal{P}_{0,r}}  P \left ( \psi_r =  1 \right )  \le \alpha \). The latter is a weaker condition that only ensures the type I error is controlled in large sample sizes.

Given tests \( \psi_1, \dots, \psi_d \), we construct the dual confidence set
\( \hat{\Theta} \coloneq \left \{  j \in [d]: \psi_j = 0 \right \}  \),
which denotes the indices that were not rejected by the corresponding tests. Our goal will then be to show that \( \hat{\Theta} \) satisfies point-wise coverage of \( \Theta \) i.e.
\begin{equation} \label{eq:pointwise-mcs-coverage} \left ( \overset{(as)}{\liminf_{n \to \infty}} \right ) \inf_{P \in \mathcal{P}} \inf_{r \in \Theta \left ( P \right )} P \left ( r \in \hat{\Theta} \right )  \ge 1-\alpha.\end{equation}

 \subsection{From MSS to MCS: General Strategies for MCS}

\label{sec:mss-to-mcs}

Interestingly, the MCS problem is strongly related to problems posed in the MSS literature. To see this, let \( \mathcal{F} \left ( \mc \right ) = \left \{ f \;|\; f: \bigcup_{m=1}^{\infty} \mathcal{Z}^m \rightarrow \mc \right \}  \) denote a collection of measurable algorithms that map from the space of sample data \( \bigcup_{m=1}^{\infty}\mathcal{Z}^m \) to a model in \( \mc \) and let \( f_{1} \in \mathcal{F} \left ( \mc_{1} \right ), \dots, f_d \in \mathcal{F} \left ( \mc_d \right ) \) represent pre-chosen model fitting algorithms that map to each model class.

We require a risk metric for \( f \in \ac \left ( \mc \right ) \) similar to that of \( R(h) \) for \( h \in \mc \). To account for risk estimation methods that utilize an algorithm to fit multiple models (e.g. cross-validation), let \( k_n > 0 \) be a sequence of positive integers specifying the number of models we will fit for each class. Define \( \left \{  \ndata_j \right \}_{j=1}^{k_n} \ \) to be a collection of \( k_n \) training sets for fitting models where \( \ndata_j \subset \ndata \) for \( j = 1,...,k_n \).  For simplicity, assume that each of these subsets have the same number of observations and set \( n_{tr} = |\ndata_1| =  \cdots = |\ndata_{k_n} | \). For any algorithm \( f \in \ac \left ( \mc \right ) \), define the risk of \( f \) to be \( R_{n}(f) = \sum_{j=1}^{k_{n}} R \left ( f \left ( \ndata_j \right ) \right ) \). In works on cross-validation, the quantity \( R_n (f) \) has been referred to as the multi-fold test error and characterizes the performance of the ensemble of models fit across the \( k_n \) datasets \citep{bayle_cross-validation_2020, austern_asymptotics_2025}. The multi-fold test error is similar to the out-of-sample error, \( \mathbb{E} \left [ R_n(f) \right ] \), which measures the expected performance of \( f \in \ac \left ( \mc \right )\) across all possible training sets of size \( n_{tr} \). Although it is possible to estimate the out-of sample error using cross-validation based approaches \citep{bates_cross-validation_2023}, such estimation is typically complex and requires strong stability conditions on the chosen loss function and algorithms \citep{luo_limits_2024}.

Previous works on MSS have devised strategies for constructing confidence sets for \( \argmin_{k \in [d]} R_n \left ( f_{k} \right ) \), which corresponds to the indices of algorithms \( f_1, \dots, f_d \) that achieve the minimal multi-fold test error. Procedures immediately applicable to estimating  \( \argmin_{k \in [d]} R_n \left ( f_{k} \right ) \) are typically based on cross-validation \citep{lei_cross-validation_2020,kissel_black-box_2023}. Other methods could be used in cases where the algorithms \( f_1, \dots, f_d \) are constant and only map to a single model (e.g. \cite{fisher_all_2019,wasserman_universal_2020,takatsu_bridging_2025,kim_locally_2025}).


Let \( \phi: \left \{ Z_1, \dots, Z_{n} \right \} \times \mathcal{F} \left ( \mc \right ) \rightarrow \left \{ 0,1 \right \}  \) denote a test function that rejects the null hypothesis \( H^{MSS}_{0,n,r}: r \in \Theta^{MSS}_{\epsilon, n} \) when \( \phi \left ( f_r \right ) = 1 \) where \( \Theta^{MSS}_{\epsilon, n} = \left \{ k \in [d] : R_n \left ( f_k \right ) \le \min_{j \in [d]} R_n \left ( f_j \right ) + \epsilon  \right \}  \) denotes the indices of algorithms that are near-optimal with margin of error \( \epsilon > 0 \) in regards to the multi-fold test error. Then, \( \phi \) appropriately controls the MSS type I error for \( r \in [d] \) if
\begin{equation} \left ( \overset{(as)}{\limsup_{n \to \infty}} \right )  \sup_{P \in \mathcal{P}}  P \left ( \phi \left ( f_r \right ) =  1 \cap  r \in \Theta^{MSS}_{\epsilon, n} \right )  \le \alpha. \label{eq:mss-test-condition}\end{equation}
It should be emphasized that \( \Theta \) is fixed, while \( \Theta^{MSS}_{\epsilon, n} \) is a random quantity that depends on the sampled training data. The term \( \epsilon > 0 \) can be viewed as a user chosen amount the error has to improve for there to exist a valid trade-off in performance between algorithms or model classes.

Throughout the remainder of this work, we make the distinction between the finite-sample and asymptotic versions of conditions (\ref{eq:test-condition}), (\ref{eq:pointwise-mcs-coverage}), and (\ref{eq:mss-test-condition}) when necessary. In cases where we do not specify the exact form of the condition, this means that the result holds in either case. In other words, it holds both when treating (\ref{eq:test-condition}), (\ref{eq:pointwise-mcs-coverage}), and (\ref{eq:mss-test-condition}) as all finite-sample or all asymptotic. The following result shows that if \( \phi \) satisfies (\ref{eq:mss-test-condition}) for all \( r \in [d] \), then we may use \( \phi \) to construct a valid test for MCS. For ease of notation, let \( \phi \left (h \right ) \) correspond to \(\phi \left ( f_h \right )\) where \( f_h: \bigcup_{m=1}^{\infty} \mathcal{Z}^{m} \rightarrow \left \{  h \right \}  \) is a constant algorithm for all \( h \in \mc \).

\begin{proposition}[Uniform MCS Test]
  \label{prop:uniform}
  Define the test \( \psi_r = I \left \{  \inf_{t \in T_r} \phi \left ( h_t \right ) = 1\right \}  \) for \( r \in [d] \). 
  Assume for some \( r \in [d] \) that \( f_r \) satisfies \begin{equation} \label{eq:rashomon-set} R_n \left ( f_{r }\right ) \le \inf_{f \in \ac \left ( \mc_r \right )}  R_n \left ( f \right ) + \epsilon \end{equation} and \( \phi \) satisfies (\ref{eq:mss-test-condition}). Then, \( \psi_r \) satisfies (\ref{eq:test-condition}). Moreover, if (\ref{eq:mss-test-condition}) and (\ref{eq:rashomon-set}) are satisfied for all \( r \in [d] \), then the confidence set \( \hat{\Theta} \) generated by \( \psi_1, \dots, \psi_{d} \) satisfies (\ref{eq:pointwise-mcs-coverage}).
\end{proposition}
In other words, Proposition \ref{prop:uniform} shows that the test \( \psi_r = I \left \{  \inf_{t \in T_r} \phi \left ( h_t \right ) = 1\right \} \) for \( r \in [d] \) satisfies the point-wise coverage property (\ref{eq:pointwise-mcs-coverage}) if there exist algorithms \( f_1, \dots, f_d \) satisfying (\ref{eq:rashomon-set}) for all \( r \in [d] \).
Notably, the computation of \( \psi_1, \dots, \psi_{d} \) does not require identifying \( f_1, \dots, f_d \); all that is required is that such algorithms satisfying (\ref{eq:mss-test-condition}) and (\ref{eq:rashomon-set}) exist. The test \( \psi_r = I \left \{  \inf_{t \in T_r} \phi \left ( h_t \right ) = 1\right \} \) is equivalent to rejecting \( \mc_{r} \) if the class does not appear to contain any well-performing models based on the MSS test \( \phi \). This approach is synonymous with recently proposed methods for universal hypothesis testing which is equivalent to the problem of MCS when \( d=2 \), \( \mc_1 \subseteq \mc_2 \), and the primary goal is to test \( H_{0,1} \) \citep{wasserman_universal_2020,dey_generalized_2025}.

The advantage of applying \( \phi \) to the entire model class is that both finite-sample and asymptotic coverage guarantees hold. The computation of such a procedure, however, may be computationally expensive. In large sample size scenarios, one would hope that we could utilize the data to select a set of nearly optimal models from each class and use that to conduct the comparison instead. Under suitable assumptions on \( f_1, \dots, f_d \), it turns out that this is a feasible approach to obtain an asymptotic result similar to Proposition \ref{prop:uniform}. 
\begin{proposition}[Selective MCS Test]
  \label{prop:selective}
  Define the test  \( \psi_r =  \phi \left ( f_{r} \right )  \) for all \( r \in [d] \). Assume for some \( r \in [d] \) that \( f_r \) satisfies
  \begin{equation} \label{cond:rashomon-convergence} R_n \left ( f_r \right ) - \inf_{t \in T_r} R \left ( h_t \right ) - \epsilon \le o_p \left ( 1 \right ) \end{equation}
  and \( \phi \) satisfies (\ref{eq:mss-test-condition}). Then, \( \psi_r \) satisfies the asymptotic version of (\ref{eq:test-condition}).  Moreover, if (\ref{eq:mss-test-condition}) and (\ref{cond:rashomon-convergence}) are satisfied for all \( r \in [d] \), then the confidence set \( \hat{\Theta} \) generated by \( \psi_1, \dots, \psi_{d} \) satisfies the asymptotic version of (\ref{eq:pointwise-mcs-coverage}).
\end{proposition}

Although condition (\ref{cond:rashomon-convergence}) is a strong requirement, it allows us to accurately underestimate the optimal risk \( \inf_{t \in T_r} R \left ( h_t \right ) \) in large samples. The condition can generally be satisfied in learnable problems when \( f_r \left ( \ndata \right ) \) computes an approximate empirical risk minimizer (ERM) across \( \mc_r \) i.e.
\[ f_r \left ( \ndata \right ) \in \left \{ h \in \mc_r \; | \; n^{-1} \sum_{Z \in \ndata} \ell \left ( h, Z \right ) \le  \arg \min_{h \in \mc_r} n^{-1} \sum_{Z \in \ndata} \ell \left ( h, Z \right ) + o \left ( n^{-\beta} \right )\right \} \]
where  \(\beta  > 0 \) controls the rate of convergence. A problem is often considered to be learnable if the empirical risks of models in the  model class converge to their population risk uniformly \citep{vapnik_nature_2000,shalev-shwartz_learnability_2010}. Such uniform convergence can be written as
\begin{equation} \label{eq:unif-convergence} \sup_{t \in T_{r}} \left | R \left ( h_{t} \right ) - \frac{1}{n} \sum_{i=1}^{n} \ell \left (h_{t}, Z_i \right ) \right | \overset{p}{\to} 0  \end{equation}
under the data generating process \( P \). It has been shown that (\ref{eq:unif-convergence}) is implied for a model class \( \mc_{r} \) whenever \( \mc_{r} \) is finite, has finite VC dimension, or has finite fat-shattering dimension \citep{vapnik_nature_2000}. Recent work has also shown that under sufficient stability assumptions, there exist strategies besides almost-ERMs satisfying (\ref{cond:rashomon-convergence}) \citep{shalev-shwartz_learnability_2010}.

\section{Proposed Methods}
\label{sec:proposed-methods}


In Section \ref{ch:mcs-construction}, we showed that if there exists a MSS test \( \phi \) satisfying (\ref{eq:mss-test-condition}), then by Propositions \ref{prop:uniform} or \ref{prop:selective}, we may use \( \phi \) to a construct a MCS test for the model class \( \mc_r \). As a result, in this section, we detail and extend current data splitting based approaches for MSS that generally satisfy these conditions under mild assumptions.


To motivate the proposed methodology, for all \( r \in [d] \), let \( s_r \in [d] \backslash r \) be any other class index besides \( r \).  Let \( \ndata_{-j} = \ndata \backslash \ndata_j \) denote the data held out from training for each data split \( j = 1,..,k_n \). In addition, set \( n_{te} = n - n_{tr}\) to be the number of held out observations.

Observe that the hypothesis \( H^{MSS}_{0,n,r} \) implies \( R_n \left ( f_r \right ) \le \min_{j \in [d]} R_n \left ( f_j \right ) + \epsilon \). So, we would hope to use the data to estimate \( R_n(f_r) - \min_{j \in [d]} R_n \left ( f_j \right )  - \epsilon \). However, we do not have access to \( \min_{j \in [d]} R_n \left ( f_j \right ) \). To alleviate this issue, note
that when \( H^{MSS}_{0,n,r} \) is true, \( R_n \left ( f_r \right ) \le  R_n \left ( f_{s_{r}} \right ) + \epsilon \) and so we can instead test if \( R_n(f_r) > R_n(f_{s_r}) + \epsilon \). We are able to estimate the expected risk difference  \( R_n(f_r) - R_n(f_{s_{r}})\), by the sample average \(\bar{R}_n \left (f_{r}, f_{s_{r}} \right ) \) where

\[ \bar{R}_n \left (f_r, f_{s_r} \right ) = \frac{1}{k_n n_{te}} \sum_{j=1}^{k_n} \sum_{Z \in \ndata_{-j}} \nabla_{n,j} \left ( f_{r}, f_{s_r}, Z \right ) \]

and
\[ \nabla_{n,j}(f_{r}, f_{s_r}, Z) =  \ell \left ( f_r \left ( \ndata_j \right ), Z \right ) - \ell \left ( f_{s_r} \left ( \ndata_j \right ), Z \right ) \text{ for } j = 1,...,k_n. \]

\subsection{Studentization}
\label{sec:studentization}

A collection of recently proposed approaches for MSS are based on sample splitting and self normalization strategies \citep{takatsu_bridging_2025,kim_locally_2025}. In this section, we show that such studentized procedures also allow for the construction of valid methods for MCS.

To showcase the approach, let
\[  \bar{\sigma}_n \left (f_r, f_{s_r} \right ) = \frac{1}{k_n} \sum_{j=1}^{k_n} \frac{1}{n_{te} - 1} \sum_{Z \in \ndata_{-j}} \left \{  \nabla_{n,j} \left (f_r, f_{s_r}, Z \right ) - \frac{1}{n_{te}} \sum_{Z \in \ndata_{-j}} \nabla_{n,j} \left ( f_r, f_{s_r}, Z \right ) \right \} ^{2} \]
be an average of the sample variances of the differences in loss between the fitted models. The following result shows that when there is a single data split (\( k_{n}=1 \)), a corresponding test dependent on central limit theory satisfies (\ref{eq:mss-test-condition}) in large samples.

\begin{theorem}[Validity of \( \phi_{CLT} \) when \( k_n = 1 \)]
  \label{thm:sui-holdout}
  Set \( k_n = 1 \). Suppose \( n_{te} \) is an increasing sequence such that \( n_{te} \rightarrow \infty \). Let \( \mu_n =  R_n(f_r) - R_n (f_{s_r})\) and \(\sigma_n^2 = \text{Var}_{Z_{0}} \left [  \nabla_{n,1}(f_r, f_{s_r}, Z_{0}) \right ] \). So long as the sequence \( \left ( \nabla_{n,1}(f_r, f_{s_r}, Z_{0})  - \mu_n \right )^{2} / \sigma_n^2  \) is uniformly integrable, the test
\begin{equation} \label{eq:sui} \phi_{CLT} \left ( f_r \right ) = I \left \{  \bar{R}_n \left (f_r, f_{s_r} \right ) > k^{-1/2}_n n_{te}^{-1/2} \bar{\sigma}_n \left ( f_r, f_{s_r} \right ) \left (\Phi^{-1}(1-\alpha) + \epsilon \right ) \right \}  \end{equation}
satisfies the asymptotic version of (\ref{eq:mss-test-condition}).
\end{theorem}

The uniform integrability condition ensures that the differences in losses are well-behaved and hold if \( \sup_{n} \mathbb{E} \left [ \left |   \nabla_{n,1}  \left (f_r, f_{s_r}, Z_0 \right ) / \sigma_n \right |^{\beta} \right ] < \infty \) for some \( \beta > 2 \) (which occurs under bounded loss) and does not converge to a degenerate distribution \citep{bayle_cross-validation_2020}.

Computation of \( \inf_{t \in T_r } \phi_{CLT} \left (h_{t} \right ) \) is difficult since we have to account for the sample standard deviation estimate \( \bar{\sigma}_n \left ( f_r, f_{s_r} \right ) \). This makes Proposition \ref{prop:selective} a more feasible solution to extending the test \( \phi_{CLT} \) to MCS than Proposition \ref{prop:uniform}. Recall that to utilize Proposition \ref{prop:selective}, we would require \( R_n \left (f_r  \right )- \inf_{t \in T_r} R \left ( h_t \right ) - \epsilon \le o_{p} \left ( 1 \right ) \). In other words, the multi-fold test error of \( f_r \) must eventually be \( \epsilon \)-close to the optimal risk of models in the class \( \mc_r \). As noted in Section \ref{sec:mss-to-mcs}, taking \( f_r \) to be an (almost)-ERM is one way to satisfy this condition. For most stable algorithms, the size of the training sets \( n_{tr} \) typically controls the rate of convergence of \( R_n \left ( f_r \right ) \).  If the size of the training sets is small relative to the size of the test sets, then the event \( R_n \left (f_r  \right )- \inf_{t \in T_r} R \left ( h_t \right ) - \epsilon \le 0 \) may only be likely to occur in extremely large samples when \( \epsilon > 0 \) is very small. See Section \ref{sec:conservative-studentization} of the Appendix for an example of such a setting and a way to conservatively modify \( \phi_{CLT} \) to fix this potential issue. In practice with \( \phi_{CLT} \) in Section \ref{sec:experiments}, we set the size of the test set to be smaller than or equal to the size of the training set and \( \epsilon \approx 0 \).
 
Finally, it is worth noting that there is reason to expect Theorem \ref{thm:sui-holdout} could be established under more mild conditions than uniform integrability (see, e.g.\ recent work from \citet{takatsu_bridging_2025} and \citet{kim_locally_2025}). For simplicity, we utilized the basic definition of uniform integrability and our results were developed irrespective of previous central limit theorems related to this approach.

\subsection{Cross-Validation}
\label{sec:cross-validation}

In Theorem \ref{thm:sui-holdout}, we focused on the case when only a single data split is utilized (\( k_n=1 \)). Relying on only a single data split, however, may result in instability. The typical solution to reduce such variability is to take \( k_n > 1 \) and average across the results of multiple splits in the computation of \( \bar{R}_n \left (f_r, f_{s_r} \right ) \). Such an average involves a sum of \( k_n \) dependent statistics. Consequently, we cannot directly apply the classical central limit theorems to asymptotically control the type I error.

Cross-validation is widely recognized as a standard approach to achieving potentially more stable error estimates by partitioning the data into \( k_n \) folds and treating each as a hold-out set \citep{bates_cross-validation_2023}. Without loss of generality, for the remainder of this section, we will assume that \( k_n > 1\) evenly divides \( n \) and define the folds sequentially. Let the \( j^{th} \) test fold for \( j \in [k_n] \) be defined as \( \ndata_{-j} = \left \{ Z_{(j-1)n_{te} +1}, \dots, Z_{jn_{te}} \right \} \) where \( n_{te} = \frac{n}{k_n} \). Then, \( \bar{R}_n \left ( f_r, f_{s_r} \right ) \) may be interpreted as the cross-validation estimate of \( R_n \left ( f_r \right ) - R_n \left ( f_{s_r} \right ) \).

A collection of recent works have shown that \( \bar{R}_n \left (f_r, f_{s_r} \right ) \) will be asymptotically normal as long as the model fitting algorithms, \( f_r \) and \( f_{s_r} \), satisfy sufficient notions of loss stability \citep{austern_asymptotics_2025,bayle_cross-validation_2020}. In this work, we consider a weak form of loss stability which ensures that the standardized cross-validation estimate converges asymptotically to a standard normal centered around \( R_n \left ( f_r \right ) - R_n \left ( f_{s_r} \right ) \) \citep{bayle_cross-validation_2020}. To define this notion of loss stability, let \( \ndatai \) represent \( \ndata \) but with \( Z_i \) replaced by \( Z'_{0} \) where \( Z_0' \iid P \). For any function \( g: \bigcup_{m=1}^{\infty} \mathcal{Z}^m \times \mathcal{Z} \rightarrow \mathbb{R} \), the loss stability of \( g \) for \( m \) training observations is defined as 
  \[  \gamma_{m}^{loss}(g) = \frac{1}{m} \sum_{i=1}^{m} \mathbb{E} \left [ \left ( g' \left (\mdata, Z_0 \right ) - g' \left (\mdatai, Z_0 \right ) \right )^{2}  \right ]   \]
  where
  \[   g' \left ( \mdata, Z_0 \right ) = g \left ( \mdata, Z_0 \right ) - \mathbb{E} \left [ g \left ( \mdata, Z_0 \right ) \bigg | \mdata \right ] . \]

Beyond stability of the model fitting algorithms, we will also require integrability conditions similar to those of Theorem \ref{thm:sui-holdout}, but for the average loss difference across all possible training folds of size \(n_{tr} = n - n_{te} \). Let
\[ \bar{\nabla}_{n}(z) = \mathbb{E}_{\ndata_{1}} \left [ \nabla_{n,1} \left ( f_r, f_{s_r}, z \right ) \right ]  \]
 represent such an average. Under the assumption of loss stability, the result of Theorem \ref{thm:sui-holdout} still holds even when averaging across \( k_n  < n \) dependent statistics.

\begin{theorem}[Validity when \( k_n < n \)]
  \label{thm:sui-cv}
  Assume \( k_n < n \). Let \( \bar{\mu}_n = \mathbb{E} \left [ R_n \left ( f_r \right ) - R_n \left ( f_{s_r} \right ) \right ]\), \( \sigma_n^2 = \text{Var} \left [ \bar{\nabla}_{n} \left (  Z_0 \right ) \right ] \), and \( g \left ( \mdata, z \right ) =  \ell \left ( f_r \left ( \mdata \right ), z \right ) - \ell \left ( f_{s_{r}} \left ( \mdata \right ), z \right )\). Assume the sequence \( \left ( \bar{\nabla}_n(Z_0) - \bar{\mu}_n \right )^{2} / \sigma_n^2  \) is uniformly integrable and \( \gamma_{n_{tr}}^{loss}(g)  = o \left ( \sigma^2_n / n \right ) \). Then, the test (\ref{eq:sui}) 
satisfies the asymptotic version of (\ref{eq:mss-test-condition}).
\end{theorem}

Various learning algorithms are capable of satisfying the loss stability condition with suitable hyper-parameter choices \citep{bayle_cross-validation_2020}. A few notable examples include stochastic gradient descent, nearest neighbor, and ensemble based methods \citep{hardt_train_2016,devroye_distribution-free_1979,elisseeff_stability_2005}. Different choices of standardization or stability conditions have also been previously considered in works on cross-validation and could be a potential avenue for different results \citep{austern_asymptotics_2025,bayle_cross-validation_2020}.

\subsection{Universal Infernce}
\label{sec:ui}
Thus far, we have considered studentization strategies that yield asymptotic tests for MCS. In this section, we propose an alternative approach to creating tests that appropriately control the type I error in finite samples. The methodology is based on a MSS test referred to as universal inference (UI) \citep{wasserman_universal_2020}. The UI method was constructed as a finite-sample alternative to the asymptotic likelihood ratio test for regular statistical models and was later generalized to arbitrary loss functions and model classes \citep{dey_generalized_2025}. In the next result, we show that universal inference is a valid strategy satisfying (\ref{eq:mss-test-condition}) under a type of strong central condition.
\begin{theorem}[Validity of UI]
  \label{thm:ui}
  For some \( r \in [d] \), assume that there exists \( \bar{\omega} > 0 \) such that
  \begin{equation}\label{eq:strong-central-condition} \mathbb{E}_{Z_{0}} \exp \left [ \omega \left \{ \nabla_{n,1} \left ( f_r, f_{s_r}, Z_0 \right ) \right \}  \right] \le 1 \text{ almost surely for all } \omega \in [0, \bar{\omega}).\end{equation}
  Then, the UI test
  \begin{equation} \label{eq:ui} \phi_{UI} \left (f_r \right ) = I \left \{ \bar{R}^{\exp}_n(f_r, f_{s_r}, \omega ) - \epsilon >  \alpha^{-1}  \right \}  \end{equation}
  where
  \[ \bar{R}^{\exp}_n \left ( f_r, f_{s_r}, \omega \right ) = k_n^{-1} \sum_{j=1}^{k_n} \exp \left \{ \omega \sum_{Z \in \ndata_{-j}} \nabla_{n,j} \left ( f_r, f_{s_r}, Z \right ) \right \},    \]
  satisfies the finite-sample version of (\ref{eq:mss-test-condition}) for any \( \omega \in [0, \bar{\omega}) \).
\end{theorem}

The proof of Theorem \ref{thm:ui} follows in a similar fashion to that of \citep[Lemma 2]{dey_generalized_2025} which is related to the MCS case when \( d = 2 \) and \( k_n = 1 \). Condition (\ref{eq:strong-central-condition}) is a generalized form of the strong central condition used by \cite{dey_generalized_2025} as it reduces to their definition when \( f_{r} \) maps to  the risk minimizer of \( \mc_{r} \) and \( f_{s_{r}} \) selects models uniformly at random from \( \mc_{s_{r}}\). The primary benefit of condition (\ref{eq:strong-central-condition}) is that it does not require assuming a strong central condition on all models in \( \mc_{s_r} \), only the subset that may be selected by \( f_{s_r} \). The strong central condition has been shown to imply fast rates of convergence for learning algorithms \citep{erven_fast_2015}.  A specific case where the strong central condition holds with \( \bar{\omega} = 1 \) is when \( \mc \) is composed of regular statistical models and the loss \( \ell \) is the negative log likelihood. There are numerous other examples where this condition can be shown to hold (see \citep[Section 3]{dey_generalized_2025} or \citep[Section 2.2]{erven_fast_2015}). In situations where the correct choice of \( \omega \) is unknown, it may be possible to conservatively estimate it using general posterior calibration based on the nonparametric bootstrap \citep{dey_generalized_2025}.

Combining Theorem \ref{thm:ui} and Proposition \ref{prop:uniform} shows that if \( \inf_{t \in T_r} \phi_{UI} \left ( h_t \right ) \) is feasible to compute, then the UI procedure can exactly control the type I error rate of MCS in finite samples. On the other hand, if such optimization is not practical, then so long as we choose \( f_r \in \ac \left ( \mc_r \right ) \) well enough such that (\ref{cond:rashomon-convergence}) is satisfied, then Proposition \ref{prop:selective} shows that UI will remain valid for large samples.



\section{Alternative Approaches}
\label{sec:alternative-approaches}

In this section, we provide a brief overview of alternative methods from the MSS literature that may be applied to the problem of MCS. These methods yield both finite sample and asymptotic valid tests for (\ref{eq:mss-test-condition}). In contrast to those reported in Section \ref{sec:proposed-methods}, these approaches typically require more stringent assumptions on the model classes and algorithms of interest or are more computationally expensive.

\textbf{MCS for classes of  regular statistical models:} When the model classes \( \mc_1, \dots, \mc_d \) are composed of regular statistical models and the loss function corresponds to the negative log-likelihood, the asymptotic likelihood ratio test may be used to conduct MCS. A likelihood ratio based methodology has already been proposed to conduct valid MCS with collections of appropriately nested model classes \citep{zheng_model_2019,li_model_2019}. Another line of work proposes a similar method for model classes composed of linear mixed models \citep{jiang_fence_2008}. Furthermore, when the classes are composed of gaussian linear models, the likelihood ratio test is equivalent to conducting an F-test \citep[Proposition 6.1.2]{bickel_mathematical_2015} and provides exact coverage. Similar methodology relying on the F-test applicable to MCS for nested gaussian linear model classes has also been suggested \citep{ferrari_confidence_2015}. See Section \ref{sec:stat-models} of the Appendix for a more detailed account of how to construct a valid MCS test based on the likelihood ratio test.

\textbf{Discrete argmin inference on algorithms and models:} There exist MSS cross-validation based methods that satisfy the asymptotic version of (\ref{eq:mss-test-condition}) \citep{lei_cross-validation_2020, kissel_black-box_2023}. In contrast to the proposed methods of Section \ref{sec:proposed-methods}, however, these approaches are more computationally expensive and require sub-weibull stability assumptions as they investigate the estimated joint distribution of the differences in loss between every pair-wise combination of the algorithms \( f_1, \dots, f_d \). Another group of methods from the field of discrete argmin inference satisfy the asymptotic version of (\ref{eq:mss-test-condition}) when the algorithms \( f_1, \dots, f_d \) map to a fixed model  \citep{hansen_model_2011, kim_locally_2025, zhang_winners_2025}. With additional modification, these methods are able to be applied to MCS. A simple example is the studentization method of Section \ref{sec:studentization}, as it relies on the same data splitting and self normalization strategies of \cite{kim_locally_2025}.

\textbf{Concentration inequalities:} Universal inference falls under a general class of approaches based on concentration inequalities. \cite{takatsu_bridging_2025} showcase a valid construction of another applicable concentration inequality using the one-sided empirical Bernstein inequality. Various forms of other inequalities that do not utilize data splitting can be found in the Probably Approximately Correct (PAC) bound literature \citep{valiant_theory_1984, vapnik_nature_2000,mohri_foundations_2018, elisseeff_stability_2005,shalev-shwartz_learnability_2010,feldman_generalization_2018,dziugaite_computing_2017}. Generally, these methods could be used to conduct finite-sample inference for MCS when combined with Proposition \ref{prop:uniform}. These approaches, however, typically require stringent assumptions on the complexity of the model classes or algorithms of interest, or exponential moment inequalities similar to the strong central condition. Under weaker but similar conditions, concentration inequalities can also be used to construct methods satisfying the asymptotic version of (\ref{eq:mss-test-condition}). For example, procedures for constructing empirical rashomon sets can be used for MCS if there is a constraint on the complexity of the model classes \citep{fisher_all_2019}. Further discussion of these methods is provided in Appendix \ref{sec:ci}.


\section{Numerical Experiments}

\label{sec:experiments}

In this section, we examine the power of the proposed methodology for MCS in simulated and real data scenarios. To conduct MCS, the studentized test \( \phi_{CLT} \) is combined with the testing strategy outlined in Proposition \ref{prop:selective}. Such an approach directly utilizes \( \phi_{CLT} \left ( f_r \right ) \) to test \( H_{0,r}: r \in \Theta \) for any \( r \in [d] \). We use \( \phi_{CLT}^{Holdout} \) and \( \phi_{CLT}^{CV} \) to distinguish between the holdout \( (k_n = 1) \) and cross-validation (\(k_n < n \)) versions of the studentized test discussed in Sections \ref{sec:studentization} and \ref{sec:cross-validation}, respectively. In the holdout case, a \( 50\%/50\% \) split of the data is used to construct the train/test sets. In the cross-validation case, the number of folds is set to \( k_n=10 \). Any performance improvement will be considered significant, so we set \( \epsilon \approx 0\). 

Where applicable, the test \( \phi_{CLT} \) is compared to the UI approach detailed in Section \ref{sec:ui}. The notation \( \phi_{UI-\omega_{or}} \) corresponds to an application of UI to MCS by combining \( \phi_{UI} \) with the testing procedure of Proposition \ref{prop:uniform}. In other words, \( \phi_{UI-\omega_{or}} \) represents the classical approach to hypothesis testing with UI where we reject \( \mc_r \) if \( \inf_{t \in T_{r}} \phi_{UI} \left ( h_{t}  \right ) = 1 \) \citep{wasserman_universal_2020}. Theorem \ref{thm:ui} and Proposition \ref{prop:uniform} show that as long as \( \omega > 0 \) is chosen such that condition (\ref{eq:strong-central-condition}) holds, then \( \phi_{UI-\omega_{or}} \) allows for valid inference in finite samples. In the following experiments, there is not a clear choice for \( \omega \in (0, \bar{\omega}] \). A nonparametric bootstrap procedure was recently built for estimating \( \bar{\omega} \) \citep{dey_generalized_2025}. However, this approach is computationally expensive to replicate across many simulations, has been shown to generally under-estimate \( \bar{\omega} \), and requires that an ERM must be able to be computed for \( \mc_r \) \citep{dey_generalized_2025}. So, instead, to examine the potential power of UI in the following simulations, we use an oracle \( \omega_{or} \) that is estimated to appropriately control type I error. For testing \( H_{0,r} \), in an approach similar to \citep[Algorithm 2]{dey_generalized_2025}, \( \omega_{or} \) was chosen by drawing samples of size \( n \) from the data generating process under \( H_{0,r} \), computing \( \phi_{UI} \left ( f_r \right ) \) on each sample for varying \( \omega \in [0,1] \), then choosing \( \omega_{or} \) to be the largest value for \( \omega \) such that the type I error across the simulations was below \( \alpha \).


\subsection{Nonlinear Regression}
\label{sec:nonlinear-regression}

A primary motivation for MCS is to be able to compare simpler classes of models to complex alternatives. First, we investigate if the proposed methodology is able to properly identify classes capable of modeling nonlinear effects in a nonlinear setting. Assume that \( Z_i = (Y_i, X_i^{\top}) \in \mathbb{R} \times \mathbb{R}^{5} \) where \( X_{i,j} \iid U[0,1] \) for \( i = 1,..,n \),  \( j = 1,...,5 \), and \( U[0,1] \) denotes the uniform distribution on \( [0,1] \). Let
\[ Y_{i} = 10 \gamma \sin \left ( \pi X_{i,1} X_{i,2} \right ) + 20 \gamma \left ( X_{i,3} - 0.05 \right )^2 + 10 X_{i,4} + 5 X_{i,5} + \epsilon_{i}\]
where \(\epsilon_i \iid N(0,10)\) and \(\gamma \ge 0\). We refer to this data generating process as the MARS regression model based on previous works that make use of it \citep{friedman_multivariate_1991, mentch_randomization_2020}. The term \( \gamma \ge 0 \) controls the degree of non-linearity and interactive components in the process. When \( \gamma = 0 \), the MARS model is the same as a gaussian linear model with two features. We consider a class of additive linear models, \( \mc_1 \), and a class of random forest regressors \( \mc_2 \). Specifically, a model \( h \in \mc_1 \) makes predictions of the form \( h \left ( Z_{i} \right ) = \beta^{\top} X_i \) for some \( \beta \in \mathbb{R}^{5} \). In contrast, a model \( h \in \mc_2 \) represents an ensemble of decision trees, with predictions for \( Z_i \) given by the average output of those trees when applied to the features \( X_i \). The algorithms \( f_1 \) and \( f_2 \) correspond to the OLS estimator of \( \mc_1 \) and the sklearn \texttt{randomforestregressor} algorithm with the argument \(n_{estimators} = 30\), respectively. Lastly, the loss function, \( \ell \), is taken as squared error loss: \( \ell \left (h, Z_{i} \right ) = \left ( h \left (Z_{i} \right ) - Y_i \right )^{2} \). The goal of the experiment is to test whether or not an additive linear model is capable of achieving the optimal risk (i.e. in the notation of Section \ref{ch:mcs-construction} the goal is to test \( H_{0,1}: 1 \in \Theta \)).  In the tests, we compare \( f_1 \) to \( f_{s_1} \) where \( s_1 = 2 \). When \( \gamma = 0 \), \( H_{0,1} \) is true. On the other hand, when \( \gamma > 0 \), \( 1 \not \in \Theta \) due to the nonlinear effects. For the \( \phi_{UI-\omega_{or}}, \phi^{Holdout}_{CLT},\) and \( \phi^{CV}_{CLT} \) methods, the results are displayed in Figure \ref{fig:mars-sim} for \( 0 \le \gamma \le 6 \) and sample sizes \( n = 200, 500, \text{ and } 1,000 \). All three methods are capable of detecting the nonlinear signal for large \( \gamma  > 0\). Notably, the \( \phi_{CLT}^{CV} \) approach attain higher power levels than UI with \( \omega_{or} \) for most \( \gamma > 0 \).

     \begin{figure}[t]
	\centering
	\includegraphics[width=\textwidth]{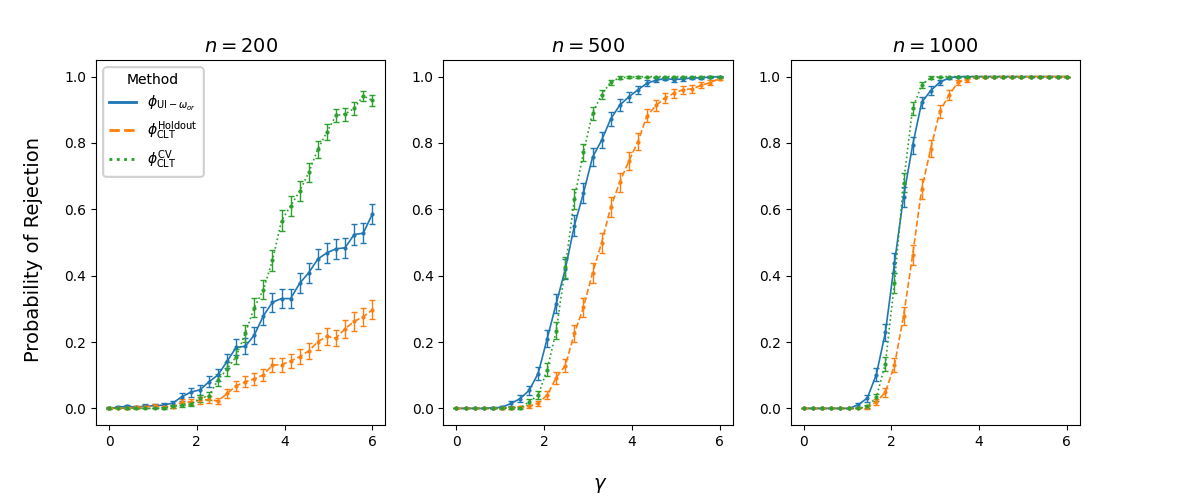}
	\caption{Rejection probabilities estimated across 1,000 simulations for various tests in the nonlinear regression setting with nominal level \( \alpha = 0.05 \). The solid, dashed, and dotted lines correspond to the tests \( \phi_{UI-\omega_{or}}, \phi^{Holdout}_{CLT},\) and \( \phi^{CV}_{CLT} \), respectively. The bars corespond to estimated 95\% confidence intervals for the power. We estimate the oracle choice of \( \omega_{or} \) from the data generating process when \( \gamma = 0 \).}
	\label{fig:mars-sim}
      \end{figure}

      \subsection{Feature Selection}
      \label{sec:feature-selection}

It is also possible to use MCS to compare nested model classes representing different combinations of features. In contrast to the previous section, the simpler classes in the following experiment will be those that utilize a smaller number of the features. Consider the regression case where for \( i = 1,...,n \), \( Z_i = (Y_i, X_i) \in \mathbb{R} \times \mathbb{R}^{6}\), \( X_i \iid N \left ( 0_6, \Sigma \right ) \), \( \epsilon_i \iid N(0, \sigma^2) \), and \( Y_i = X_{i,1} + X_{i,3} + X_{i,5} + \epsilon_i. \)
The covariance matrix, \( \Sigma \in \mathbb{R}^{6 \times 6} \), has entries \( \Sigma_{i,j} = \rho^{|i - j|} \) with \( \rho = 0.35 \). In line with previous works, we choose the noise level \( \sigma^{2} \) based on a corresponding signal to noise ratio level \( \nu \) where \( \sigma^2 = \frac{\beta^{\top} \Sigma \beta}{\nu} \) with \( \beta^{\top} = (1,0,1,0,1,0) \) \citep{hastie_best_2020, mentch_randomization_2020}. In this setting, there are \( 2^6-1 = 63 \) combinations of covariates we could choose to include in a model. Let \( \ell \) again correspond to the squared error loss. Let the model classes \( \mc_1, \dots, \mc_{63} \subset \mc \) correspond to the 63 unique covariate inclusion combinations. Assume \( \mc_{63} \) corresponds to the class of models that includes all the features. Provided the simulated data and tests \( \phi_{r} \) for \(  r \in [d] \), we construct the confidence set \( \hat{\Theta} = \left \{ r \in [62] : \phi_r = 0 \right \} \). Let \( r^{*} \) correspond to the model class \( \mc_{r^{*}} \) that includes only the 1st, 3rd, and 5th covariates. In Figure \ref{fig:feature-selection}, for \( 0.25 \le \nu \le 4 \), we estimate the miss-coverage probability of \( r^{*} \not \in \hat{\Theta} \), average size of \( \hat{\Theta} \), and average uniform coverage rate of \( \hat{\Theta} \supseteq \Theta  \). We consider treating \( \mc_1, \dots, \mc_{63} \) as either classes of additive linear models or random forests with \( 30 \) estimated trees. For each model type and \( r \in [d] \), let \( f_r \) correspond to fitting models in \( \mc_r \) using the sklearn \texttt{linearregression.fit} or \texttt{randomforestregressor.fit} algorithm, respectively. For each \( r \in [d] \), we compare \( f_r \) to \( f_{s_r} \) where \( s_r = 63 \). In other words, we compare each nested class to the most complex class. In the case of the random forest model classes, only \( \phi^{Holdout}_{CLT} \) and \( \phi^{CV}_{CLT} \) are considered since \( \inf_{t \in T_r} \phi_{UI} \left ( h_{t} \right ) \) is unable to be computed for any \( r \in [d] \). The results in Figure \ref{fig:feature-selection} show that all procedures appear to attain similar size model class confidence sets \( \hat{\Theta} \). For \( \nu \le 2 \), the type I errors of \( \phi_{CLT}^{Holdout} \) and \( \phi_{CLT}^{CV} \) rise above the significance level in the random forest case. This is possible for both methods as the random forest(s) selected by \( f_r \) for \( r \in [d] \) may not be near-optimal (i.e. lie in the Rashomon set) in small sample sizes.

\begin{figure}[ht!]
  \includegraphics[width=\textwidth]{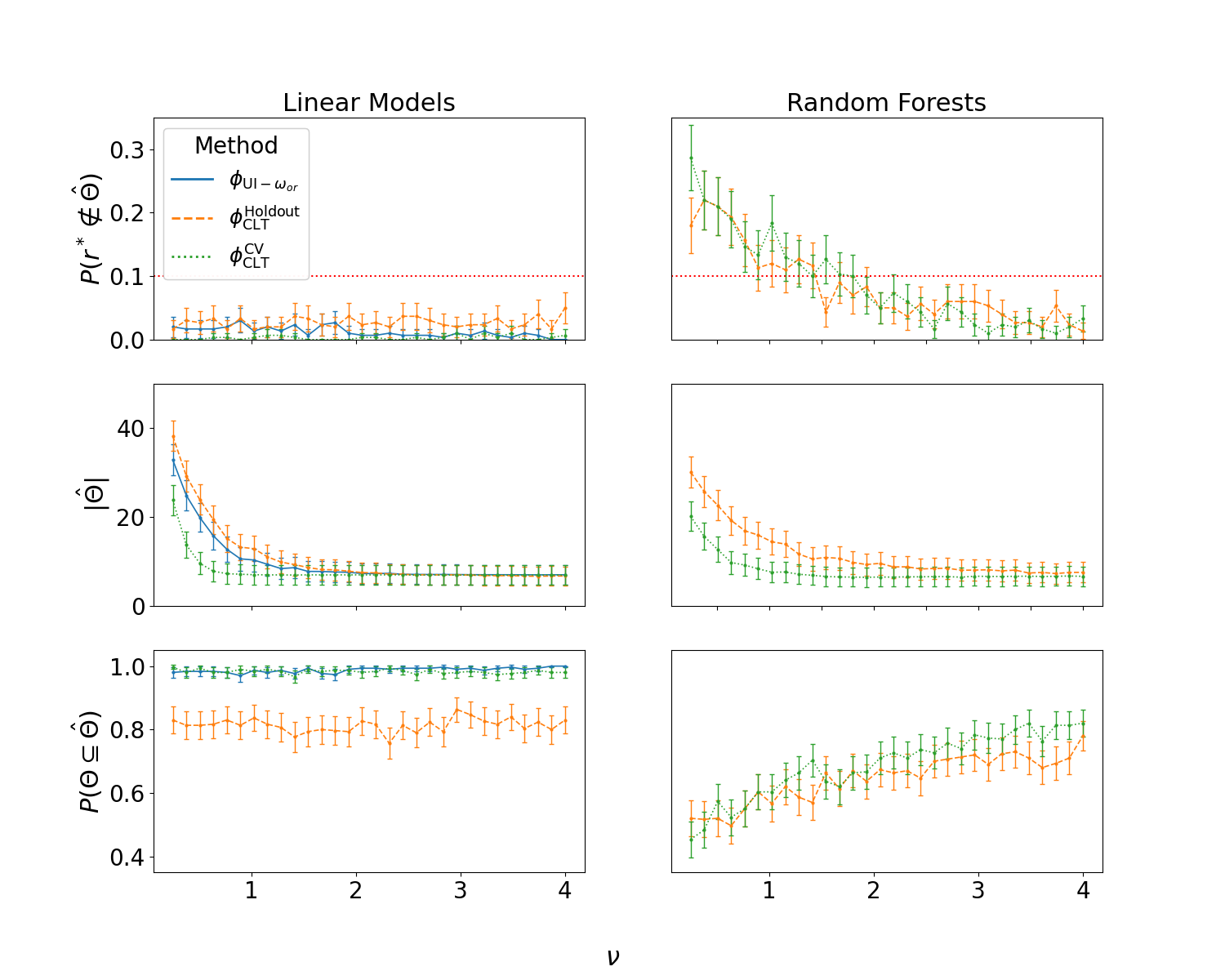}
  \caption{Feature selection performance of MCS with sample size \( n = 200 \) and \( \alpha = 0.1 \) estimated across 300 simulations. Each column corresponds to setting the classes \( \mc_1, \dots, \mc_{63} \) to different types of models. The first, second, and third row depicts the estimated miss-coverage probability of \( r^{*} \not \in \hat{\Theta} \), average size of \( \hat{\Theta} \), and average uniform coverage rate of \( \hat{\Theta} \supseteq \Theta  \) for varying levels of \( \nu \). The red, dotted line represents the significance level. The bars correspond to estimated 95\% confidence intervals. For each level \( \nu \), we estimate the oracle choice of \( \omega_{or} \) for testing \( H_{0,r^{*}}: r^{*} \in \Theta \).}
  \label{fig:feature-selection}
\end{figure}

\subsection{MCS in the Presence of Noise}
\label{sec:mcs-additive-noise}

A recent line of works have shown that adding noise to the data generating process leads to implicit regularization in classification and regression problems \citep{semenova_existence_2022,semenova_path_2023,boner_using_2024,mentch_randomization_2020}. In other words, as more noise is added to the underlying process, these results suggest that simpler models will begin to perform on par with more complex alternatives. In MCS, such an effect would appear as a loss of power, since the optimal models in each class would perform as well as any other under sufficient additive noise. We experimentally investigate this claim using the \( \phi^{Holdout}_{CLT} \) and \( \phi^{CV}_{CLT} \) methods.

Similar to Section \ref{sec:nonlinear-regression}, consider the model classes \( \mc_1 \) and \( \mc_2 \) which represent a class of additive linear models and random forests, respectively. To investigate the effect of additive noise, a collection of real-world datasets from Section \ref{sec:real-world-data} where \( \phi^{CV}_{CLT} \) rejects \(H_{0,1}: 1 \in \Theta \) are used. Two settings are considered: regression and classification. In the regression case, the algorithms \( f_1 \) and \( f_2 \) correspond to the OLS estimator of \( \mc_1 \) and the sklearn \texttt{randomforestregressor} algorithm with the argument \(n_{estimators} = 100\), respectively. In the classification case, the same algorithms represent the sklearn \texttt{logisticregression} and \texttt{randomforestclassifier} algorithms with the argument \(n_{estimators} = 1,000\). Given a dataset \( \ndata =  (Y_i, X_i)_{i=1}^n \), we generate a new dataset \( \ndata_{\rho} = (Y^{\rho}_i, X_i)_{i=1}^n \) where for \( i = 1,...,n \), in the classification case, \( Y_i^{\rho} \) is flipped independently with probability \( \rho \), and in the regression case, \( Y^{\rho}_i = Y_i + 4 \rho \epsilon_i \) where \( \epsilon_i \iid N(0,\hat{\sigma}) \) with \( \hat{\sigma} \) being the sample variance of \( Y \). Following similar experimental work, in the simulations that follow we consider \( 0 \le \rho \le .25 \) \citep{semenova_path_2023}. In the regression setting, \( \rho \) is multiplied by \( 4 \) so that the noise level of each \( \epsilon_i \) varies from \( 0 \) to \( \hat{\sigma} \). In the classification and regression cases, we let \( \ell \) correspond to the log loss and squared error loss, respectively. Figure \ref{fig:additive-noise} shows the results for the \( \phi^{Holdout}_{CLT} \) and \( \phi^{CV}_{CLT} \) methods.  The results validate previous works as we do indeed see a substantial loss in power once the noise level is sufficiently high.


\begin{figure}[t]
  \includegraphics[width=\textwidth]{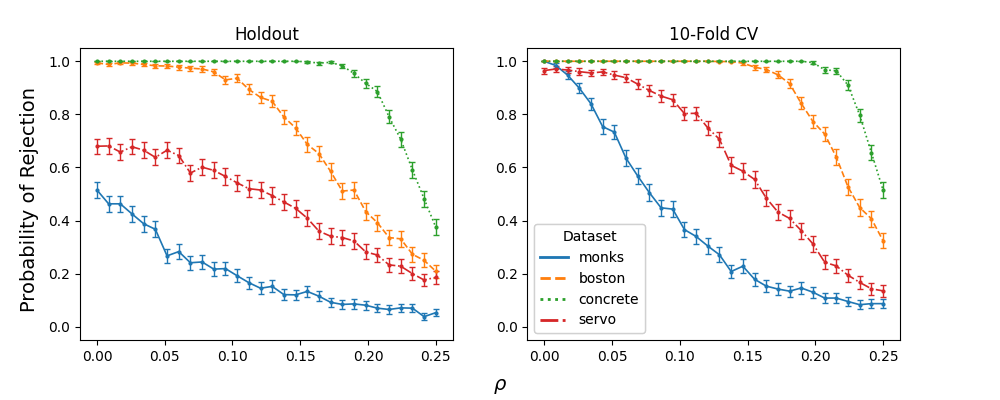}
  \caption{Effects of additive noise on \( \phi_{CLT} \) at nominal level \( \alpha = 0.1 \) estimated across 1,000 simulations. The left and right plots correspond to the estimated probability that \( \phi_{CLT}^{Holdout} \left ( f_{1 } \right ) =1\) and \( \phi_{CLT}^{CV} \left ( f_1 \right ) =1\), respectively.  The solid, dashed, dotted, and solid/dotted lines correspond to the datasets monks, boston, concrete, and servo, respectively, which are discussed further in Section \ref{sec:real-world-data}.}
  \label{fig:additive-noise}
\end{figure}

\subsection{Real-World Data}
\label{sec:real-world-data}

In this section, \( \phi^{CV}_{CLT} \) is applied to a diverse collection of real-world datasets. The goal is to test the same scenario as that of Section \ref{sec:mcs-additive-noise}. Namely, if \( H_{0,1}: 1 \in \Theta \) is true when \( \mc_1 \) and \( \mc_2 \) represent classes of additive linear models and random forests, respectively. Table \ref{tab:rl-data-tests} displays the results of the MCS test on each dataset.

Table \ref{tab:rl-data-tests} shows that when applied to real-world data, the \( \phi^{CV}_{CLT} \) test produces a range of results on both regression and classification datasets. In Figure \ref{fig:rl-cv-results} of the Appendix, the cross-validation errors of models selected using \( f_1 \) and \( f_2 \) are compared. Notably, the corresponding statistic in Table \ref{tab:rl-data-tests} is large when the random forest cross-validation errors are on average much larger than that of the linear models. Table \ref{tab:rl-data-tests} and Figure \ref{fig:rl-cv-results} also demonstrate that these results do not appear to follow a standard pattern nor obviously correlate with dataset size or dimension.

\begin{table}[p]
\centering
\begin{tabular}{l l c c c l}
\hline
Dataset & Task & n & p & \( \bar{R}_n \left ( f_1, f_{s_1} \right ) \) & Decision \\
\hline
Adult Income \citep{barry_becker_adult_1996} & C & 30162 & 96 & -4.72 & Fail \\
German Credit \citep{hofmann_statlog_1994}& C & 1000 & 49 & 1.35 & Reject \\
COMPAS \citep{larson_how_2016} & C & 6172 & 9 & -7.87 & Fail \\
Telco \citep{cognos_analytics_telco_2018} & C & 7043 & 8 & -12.16 & Fail \\
Monks \citep{wnek_monks_1993} & C & 124 & 12 & 4.71 & Reject \\
Boston Housing \citep{harrison_hedonic_1978} & R & 506 & 13 & 6.06 & Reject \\
Auto \citep{quinlan_auto_1993} & R & 392 & 9 & 5.22 & Reject \\
Bike Sharing \citep{fanaee-t_event_2014} & R & 731 & 12 & -0.17 & Fail \\
Abalone Age \citep{waugh_extending_1995}& R & 4177 & 10 & 1.64 & Reject \\
Concrete Strength \citep{yeh_modeling_1998}& R & 1030 & 9 & 14.92 & Reject \\
CPU \citep{ein-dor_attributes_1987} & R & 209 & 36 & 14.00 & Reject \\
CSM  \citep{ahmed_using_2015}& R & 187 & 24 & 3.84 & Reject \\
Facebook Metrics \citep{moro_predicting_2016}& R & 499 & 15 & -3.65 & Fail \\
Parkinsons \citep{athanasios_tsanas_parkinsons_2009} & R & 5875 & 21 & -20.96 & Fail \\
Servo System \citep{quinlan_combining_1993}& R & 167 & 11 & 3.43 & Reject \\
Melting Point \citep{bergstrom_molecular_2003}& R & 274 & 1143 & 5.86 & Reject \\
\hline
\end{tabular}
\caption{ Overview of results of \( \phi^{CV}_{CLT} \) when applied to a collection real-world regression and classification datasets. In the task column, C and R stand for classification and regression, respectively. The columns n, p, and \( \bar{R}_n \left ( f_1, f_{s_1} \right ) \) represent the sample size, number of features, and test statistic.}
\label{tab:rl-data-tests}
\end{table}


\section{Conclusion}
\label{sec:conclusions}


This work introduces a general framework for constructing tests for MCS from simpler methods. The proposed framework incorporates current uniform testing approaches as well as more efficient selective versions. Tests based on studentization and UI are shown to satisfy the requirements of this framework under mild assumptions.  Experimental results provide evidence that the proposed methods appropriately control the type I error rates in finite samples across a range of simulated MCS scenarios while also attaining notable degrees of power. Applications on real-world datasets demonstrate that the proposed framework is not only theoretically valid but also practically effective in distinguishing instances where simple models are similarly accurate from those in which more complex ML alternatives perform significantly better.


It has become widely recognized that the application of models like random forests or neural networks often leads to increased power at the expense of interpretability. Until now, a rigorous method for detecting when such interpretability-performance trade-offs exist has not been proposed. This work introduces a methodology that enables users to formally identify these trade-offs, providing a more comprehensive understanding of the data.
By adjusting the performance rejection threshold \( \epsilon > 0 \) of the proposed tests, users may also define how much the error has to significantly improve for there to exist a valid trade-off.

There are various MCS future directions to explore. While our focus was on point-wise coverage (\ref{eq:pointwise-mcs-coverage}), the proposed methodology could be modified so that \( \hat{\Theta} \) uniformly covers \( \Theta \) with a desired level of confidence. See Section \ref{sec:uniform-coverage} for more discussion on this topic. Additionally, we proposed sufficient conditions for constructing MCS tests based on MSS tests. It remains an open question whether milder or necessary conditions could be identified. Future research could also benefit from exploring and comparing alternative methodologies not examined in this study. Furthermore, conducting theoretical power analyses in simpler settings could provide valuable insights.

\section{Acknowledgements}

This research was supported in part by the \hideable{University of Pittsburgh Center for Research Computing and Data, RRID:SCR\_022735, through the resources provided. Specifically, this work used the H2P cluster, which is supported by NSF award number OAC-2117681}.

\section{Data Availability and Reproducibility}

Data sharing is not applicable to this article as no new data were created in this study. See \hideable{\url{https://github.com/ryanmcecil/model-class-selection}} to reproduce the results of Section \ref{sec:experiments}.













\bibliography{refs.bib}

@book{hastie_elements_2009,
  address =       {New York, NY},
  author =        {Hastie, Trevor and Tibshirani, Robert and
                   Friedman, Jerome},
  publisher =     {Springer New York},
  series =        {Springer {Series} in {Statistics}},
  title =         {The {Elements} of {Statistical} {Learning}},
  year =          {2009},
  doi =           {10.1007/978-0-387-84858-7},
  isbn =          {978-0-387-84857-0 978-0-387-84858-7},
  url =           {http://link.springer.com/10.1007/978-0-387-84858-7},
}

@article{breiman_statistical_2001,
  author =        {Breiman, Leo},
  journal =       {Statistical Science},
  month =         aug,
  note =          {Publisher: Institute of Mathematical Statistics},
  number =        {3},
  pages =         {199--231},
  title =         {Statistical {Modeling}: {The} {Two} {Cultures} (with
                   comments and a rejoinder by the author)},
  volume =        {16},
  year =          {2001},
  doi =           {10.1214/ss/1009213726},
  issn =          {0883-4237, 2168-8745},
  url =           {https://projecteuclid.org/journals/statistical-science/
                  volume-16/issue-3/Statistical-Modeling--The-Two-Cultures-
                  with-comments-and-a/10.1214/ss/1009213726.full},
}

@article{fisher_all_2019,
  author =        {Fisher, Aaron and Rudin, Cynthia and
                   Dominici, Francesca},
  journal =       {Journal of Machine Learning Research},
  number =        {177},
  pages =         {1--81},
  title =         {All {Models} are {Wrong}, but {Many} are {Useful}:
                   {Learning} a {Variable}'s {Importance} by {Studying}
                   an {Entire} {Class} of {Prediction} {Models}
                   {Simultaneously}},
  volume =        {20},
  year =          {2019},
  issn =          {1533-7928},
  url =           {http://jmlr.org/papers/v20/18-760.html},
}

@article{wasserman_universal_2020,
  author =        {Wasserman, Larry and Ramdas, Aaditya and
                   Balakrishnan, Sivaraman},
  journal =       {Proceedings of the National Academy of Sciences},
  month =         jul,
  note =          {Publisher: Proceedings of the National Academy of
                   Sciences},
  number =        {29},
  pages =         {16880--16890},
  title =         {Universal inference},
  volume =        {117},
  year =          {2020},
  doi =           {10.1073/pnas.1922664117},
  url =           {https://www.pnas.org/doi/10.1073/pnas.1922664117},
}

@article{kissel_forward_2024,
  author =        {Kissel, Nicholas and Mentch, Lucas},
  journal =       {Statistics and Computing},
  month =         feb,
  number =        {2},
  pages =         {82},
  title =         {Forward stability and model path selection},
  volume =        {34},
  year =          {2024},
  doi =           {10.1007/s11222-024-10395-8},
  issn =          {1573-1375},
  language =      {en},
  url =           {https://doi.org/10.1007/s11222-024-10395-8},
}

@article{drazen_where_2023,
  author =        {Hunter, David J. and Holmes, Christopher},
  editor =        {Drazen, Jeffrey M. and Kohane, Isaac S. and
                   Leong, Tze-Yun},
  journal =       {New England Journal of Medicine},
  month =         sep,
  number =        {13},
  pages =         {1211--1219},
  title =         {Where {Medical} {Statistics} {Meets} {Artificial}
                   {Intelligence}},
  volume =        {389},
  year =          {2023},
  doi =           {10.1056/NEJMra2212850},
  issn =          {0028-4793, 1533-4406},
  language =      {en},
  url =           {http://www.nejm.org/doi/10.1056/NEJMra2212850},
}

@article{abdar_review_2021,
  author =        {Abdar, Moloud and Pourpanah, Farhad and
                   Hussain, Sadiq and Rezazadegan, Dana and Liu, Li and
                   Ghavamzadeh, Mohammad and Fieguth, Paul and
                   Cao, Xiaochun and Khosravi, Abbas and
                   Acharya, U. Rajendra and Makarenkov, Vladimir and
                   Nahavandi, Saeid},
  journal =       {Information Fusion},
  month =         dec,
  pages =         {243--297},
  title =         {A review of uncertainty quantification in deep
                   learning: {Techniques}, applications and challenges},
  volume =        {76},
  year =          {2021},
  doi =           {10.1016/j.inffus.2021.05.008},
  issn =          {15662535},
  language =      {en},
  url =           {https://linkinghub.elsevier.com/retrieve/pii/
                  S1566253521001081},
}

@incollection{kutyniok_modern_2022,
  address =       {Cambridge},
  author =        {Berner, Julius and Grohs, Philipp and Kutyniok, Gitta and
                   Petersen, Philipp},
  booktitle =     {Mathematical {Aspects} of {Deep} {Learning}},
  editor =        {Kutyniok, Gitta and Grohs, Philipp},
  pages =         {1--111},
  publisher =     {Cambridge University Press},
  title =         {The {Modern} {Mathematics} of {Deep} {Learning}},
  year =          {2022},
  doi =           {10.1017/9781009025096.002},
  isbn =          {978-1-316-51678-2},
  url =           {https://www.cambridge.org/core/books/mathematical-aspects-
                  of-deep-learning/modern-mathematics-of-deep-learning/
                  7C3874F83A5D934E5FDC984B8457D553},
}

@article{mentch_quantifying_2016,
  author =        {Mentch, Lucas and Hooker, Giles},
  journal =       {Journal of Machine Learning Research},
  number =        {26},
  pages =         {1--41},
  title =         {Quantifying {Uncertainty} in {Random} {Forests} via
                   {Confidence} {Intervals} and {Hypothesis} {Tests}},
  volume =        {17},
  year =          {2016},
  issn =          {1533-7928},
  url =           {http://jmlr.org/papers/v17/14-168.html},
}

@article{rudin_interpretable_2022,
  author =        {Rudin, Cynthia and Chen, Chaofan and Chen, Zhi and
                   Huang, Haiyang and Semenova, Lesia and Zhong, Chudi},
  journal =       {Statistics Surveys},
  month =         jan,
  note =          {Publisher: Amer. Statist. Assoc., the Bernoulli Soc.,
                   the Inst. Math. Statist., and the Statist. Soc.
                   Canada},
  number =        {none},
  pages =         {1--85},
  title =         {Interpretable machine learning: {Fundamental}
                   principles and 10 grand challenges},
  volume =        {16},
  year =          {2022},
  doi =           {10.1214/21-SS133},
  issn =          {1935-7516},
  url =           {https://projecteuclid.org/journals/statistics-surveys/
                  volume-16/issue-none/Interpretable-machine-learning-
                  Fundamental-principles-and-10-grand-challenges/10.1214/21-
                  SS133.full},
}

@article{wexler_opinion_2017,
  author =        {Wexler, Rebecca},
  chapter =       {Opinion},
  journal =       {The New York Times},
  month =         jun,
  title =         {Opinion {\textbar} {When} a {Computer} {Program}
                   {Keeps} {You} in {Jail}},
  year =          {2017},
  issn =          {0362-4331},
  language =      {en-US},
  url =           {https://www.nytimes.com/2017/06/13/opinion/how-computers-
                  are-harming-criminal-justice.html},
}

@article{varshney_safety_2017,
  author =        {Varshney, Kush R. and Alemzadeh, Homa},
  journal =       {Big Data},
  month =         sep,
  note =          {Publisher: Mary Ann Liebert, Inc., publishers},
  number =        {3},
  pages =         {246--255},
  title =         {On the {Safety} of {Machine} {Learning}:
                   {Cyber}-{Physical} {Systems}, {Decision} {Sciences},
                   and {Data} {Products}},
  volume =        {5},
  year =          {2017},
  doi =           {10.1089/big.2016.0051},
  issn =          {2167-6461},
  url =           {https://www.liebertpub.com/doi/10.1089/big.2016.0051},
}

@article{rudin_stop_2019,
  author =        {Rudin, Cynthia},
  journal =       {Nature Machine Intelligence},
  month =         may,
  note =          {Publisher: Nature Publishing Group},
  number =        {5},
  pages =         {206--215},
  title =         {Stop explaining black box machine learning models for
                   high stakes decisions and use interpretable models
                   instead},
  volume =        {1},
  year =          {2019},
  doi =           {10.1038/s42256-019-0048-x},
  issn =          {2522-5839},
  language =      {en},
  url =           {https://www.nature.com/articles/s42256-019-0048-x},
}

@inproceedings{semenova_existence_2022,
  address =       {New York, NY, USA},
  author =        {Semenova, Lesia and Rudin, Cynthia and Parr, Ronald},
  booktitle =     {Proceedings of the 2022 {ACM} {Conference} on
                   {Fairness}, {Accountability}, and {Transparency}},
  month =         jun,
  pages =         {1827--1858},
  publisher =     {Association for Computing Machinery},
  series =        {{FAccT} '22},
  title =         {On the {Existence} of {Simpler} {Machine} {Learning}
                   {Models}},
  year =          {2022},
  doi =           {10.1145/3531146.3533232},
  isbn =          {978-1-4503-9352-2},
  url =           {https://doi.org/10.1145/3531146.3533232},
}

@article{semenova_path_2023,
  author =        {Semenova, Lesia and Chen, Harry and Parr, Ronald and
                   Rudin, Cynthia},
  journal =       {Advances in Neural Information Processing Systems},
  month =         dec,
  pages =         {3362--3401},
  title =         {A {Path} to {Simpler} {Models} {Starts} {With}
                   {Noise}},
  volume =        {36},
  year =          {2023},
  language =      {en},
  url =           {https://proceedings.neurips.cc/paper_files/paper/2023/hash/
                  0a49935d2b3d3342ca08d6db0adcfa34-Abstract-Conference.html},
}

@article{boner_using_2024,
  author =        {Boner, Zachery and Chen, Harry and Semenova, Lesia and
                   Parr, Ronald and Rudin, Cynthia},
  journal =       {Advances in Neural Information Processing Systems},
  month =         dec,
  pages =         {131824--131858},
  title =         {Using {Noise} to {Infer} {Aspects} of {Simplicity}
                   {Without} {Learning}},
  volume =        {37},
  year =          {2024},
  language =      {en},
  url =           {https://proceedings.neurips.cc/paper_files/paper/2024/hash/
                  ee1331338cdf0e4d1055304d875548df-Abstract-Conference.html},
}

@inproceedings{bayle_cross-validation_2020,
  author =        {Bayle, Pierre and Bayle, Alexandre and Janson, Lucas and
                   Mackey, Lester},
  booktitle =     {Advances in {Neural} {Information} {Processing}
                   {Systems}},
  pages =         {16339--16350},
  publisher =     {Curran Associates, Inc.},
  title =         {Cross-validation {Confidence} {Intervals} for {Test}
                   {Error}},
  volume =        {33},
  year =          {2020},
  url =           {https://proceedings.neurips.cc/paper/2020/hash/
                  bce9abf229ffd7e570818476ee5d7dde-Abstract.html},
}

@article{austern_asymptotics_2025,
  author =        {Austern, Morgane and Zhou, Wenda},
  journal =       {Annales de l'Institut Henri Poincaré, Probabilités
                   et Statistiques},
  month =         nov,
  note =          {Publisher: Institut Henri Poincaré},
  number =        {4},
  pages =         {2804--2865},
  title =         {Asymptotics of cross-validation},
  volume =        {61},
  year =          {2025},
  doi =           {10.1214/24-AIHP1488},
  issn =          {0246-0203},
  url =           {https://projecteuclid.org/journals/annales-de-linstitut-
                  henri-poincare-probabilites-et-statistiques/volume-61/issue-
                  4/Asymptotics-of-cross-validation/10.1214/24-AIHP1488.full},
}

@article{bates_cross-validation_2023,
  author =        {Bates, Stephen and Hastie, Trevor and
                   Tibshirani, Robert},
  journal =       {Journal of the American Statistical Association},
  note =          {Publisher: Taylor \& Francis \_eprint:
                   https://doi.org/10.1080/01621459.2023.2197686},
  number =        {0},
  pages =         {1--12},
  title =         {Cross-{Validation}: {What} {Does} {It} {Estimate} and
                   {How} {Well} {Does} {It} {Do} {It}?},
  volume =        {0},
  year =          {2023},
  doi =           {10.1080/01621459.2023.2197686},
  issn =          {0162-1459},
  url =           {https://doi.org/10.1080/01621459.2023.2197686},
}

@misc{luo_limits_2024,
  author =        {Luo, Yuetian and Barber, Rina Foygel},
  month =         mar,
  note =          {arXiv:2402.07388 [cs, math, stat]},
  publisher =     {arXiv},
  title =         {The {Limits} of {Assumption}-free {Tests} for
                   {Algorithm} {Performance}},
  year =          {2024},
  doi =           {10.48550/arXiv.2402.07388},
  url =           {http://arxiv.org/abs/2402.07388},
}

@article{lei_cross-validation_2020,
  author =        {Lei, Jing},
  journal =       {Journal of the American Statistical Association},
  note =          {Publisher: Taylor \& Francis \_eprint:
                   https://doi.org/10.1080/01621459.2019.1672556},
  number =        {532},
  pages =         {1978--1997},
  title =         {Cross-{Validation} {With} {Confidence}},
  volume =        {115},
  year =          {2020},
  doi =           {10.1080/01621459.2019.1672556},
  url =           {https://doi.org/10.1080/01621459.2019.1672556},
}

@misc{kissel_black-box_2023,
  author =        {Kissel, Nicholas and Lei, Jing},
  month =         nov,
  note =          {arXiv:2211.04958 [math, stat]},
  publisher =     {arXiv},
  title =         {Black-{Box} {Model} {Confidence} {Sets} {Using}
                   {Cross}-{Validation} with {High}-{Dimensional}
                   {Gaussian} {Comparison}},
  year =          {2023},
  doi =           {10.48550/arXiv.2211.04958},
  url =           {http://arxiv.org/abs/2211.04958},
}

@misc{takatsu_bridging_2025,
  author =        {Takatsu, Kenta and Kuchibhotla, Arun Kumar},
  month =         apr,
  note =          {arXiv:2501.07772 [math]},
  publisher =     {arXiv},
  title =         {Bridging {Root}-\$n\$ and {Non}-standard
                   {Asymptotics}: {Adaptive} {Inference} in
                   {M}-{Estimation}},
  year =          {2025},
  doi =           {10.48550/arXiv.2501.07772},
  url =           {http://arxiv.org/abs/2501.07772},
}

@misc{kim_locally_2025,
  author =        {Kim, Ilmun and Ramdas, Aaditya},
  month =         mar,
  note =          {arXiv:2503.21639 [math] version: 1},
  publisher =     {arXiv},
  title =         {Locally minimax optimal and dimension-agnostic
                   discrete argmin inference},
  year =          {2025},
  doi =           {10.48550/arXiv.2503.21639},
  url =           {http://arxiv.org/abs/2503.21639},
}

@article{dey_generalized_2025,
  author =        {Dey, Neil and Martin, Ryan and Williams, Jonathan P},
  journal =       {Journal of the Royal Statistical Society Series B:
                   Statistical Methodology},
  month =         oct,
  pages =         {qkaf065},
  title =         {Generalized universal inference on risk minimizers},
  year =          {2025},
  doi =           {10.1093/jrsssb/qkaf065},
  issn =          {1369-7412},
  url =           {https://doi.org/10.1093/jrsssb/qkaf065},
}

@book{vapnik_nature_2000,
  address =       {New York, NY},
  author =        {Vapnik, Vladimir N.},
  publisher =     {Springer New York},
  title =         {The {Nature} of {Statistical} {Learning} {Theory}},
  year =          {2000},
  doi =           {10.1007/978-1-4757-3264-1},
  isbn =          {978-1-4419-3160-3 978-1-4757-3264-1},
  url =           {http://link.springer.com/10.1007/978-1-4757-3264-1},
}

@article{shalev-shwartz_learnability_2010,
  author =        {Shalev-Shwartz, Shai and Shamir, Ohad and
                   Srebro, Nathan and Sridharan, Karthik},
  journal =       {Journal of Machine Learning Research},
  number =        {90},
  pages =         {2635--2670},
  title =         {Learnability, {Stability} and {Uniform}
                   {Convergence}},
  volume =        {11},
  year =          {2010},
  issn =          {1533-7928},
  url =           {http://jmlr.org/papers/v11/shalev-shwartz10a.html},
}

@inproceedings{hardt_train_2016,
  address =       {New York, NY, USA},
  author =        {Hardt, Moritz and Recht, Benjamin and Singer, Yoram},
  booktitle =     {Proceedings of the 33rd {International} {Conference}
                   on {International} {Conference} on {Machine}
                   {Learning} - {Volume} 48},
  month =         jun,
  pages =         {1225--1234},
  publisher =     {JMLR.org},
  series =        {{ICML}'16},
  title =         {Train faster, generalize better: stability of
                   stochastic gradient descent},
  year =          {2016},
}

@article{devroye_distribution-free_1979,
  author =        {Devroye, L. and Wagner, T.},
  journal =       {IEEE Transactions on Information Theory},
  month =         mar,
  note =          {Conference Name: IEEE Transactions on Information
                   Theory},
  number =        {2},
  pages =         {202--207},
  title =         {Distribution-free inequalities for the deleted and
                   holdout error estimates},
  volume =        {25},
  year =          {1979},
  doi =           {10.1109/TIT.1979.1056032},
  issn =          {1557-9654},
  url =           {https://ieeexplore.ieee.org/document/1056032},
}

@article{elisseeff_stability_2005,
  author =        {Elisseeff, Andre and Evgeniou, Theodoros and
                   Pontil, Massimiliano},
  journal =       {Journal of Machine Learning Research},
  number =        {3},
  pages =         {55--79},
  title =         {Stability of {Randomized} {Learning} {Algorithms}},
  volume =        {6},
  year =          {2005},
  issn =          {1533-7928},
  url =           {http://jmlr.org/papers/v6/elisseeff05a.html},
}

@article{erven_fast_2015,
  author =        {Erven, Tim van and Grünwald, Peter D. and
                   Mehta, Nishant A. and Reid, Mark D. and
                   Williamson, Robert C.},
  journal =       {Journal of Machine Learning Research},
  number =        {54},
  pages =         {1793--1861},
  title =         {Fast {Rates} in {Statistical} and {Online}
                   {Learning}},
  volume =        {16},
  year =          {2015},
  issn =          {1533-7928},
  url =           {http://jmlr.org/papers/v16/vanerven15a.html},
}

@article{zheng_model_2019,
  author =        {Zheng, Chao and Ferrari, Davide and Yang, Yuhong},
  journal =       {Statistica Sinica},
  note =          {Publisher: Institute of Statistical Science, Academia
                   Sinica},
  number =        {2},
  pages =         {827--851},
  title =         {Model {Selection} {Confidence} {Sets} by {Likelihood}
                   {Ratio} {Testing}},
  volume =        {29},
  year =          {2019},
  issn =          {1017-0405},
  url =           {https://www.jstor.org/stable/26705490},
}

@article{li_model_2019,
  author =        {Li, Yang and Luo, Yuetian and Ferrari, Davide and
                   Hu, Xiaonan and Qin, Yichen},
  journal =       {Biometrics},
  note =          {\_eprint:
  https://onlinelibrary.wiley.com/doi/pdf/10.1111/biom.13024},
  number =        {2},
  pages =         {392--403},
  title =         {Model confidence bounds for variable selection},
  volume =        {75},
  year =          {2019},
  doi =           {10.1111/biom.13024},
  issn =          {1541-0420},
  language =      {en},
  url =           {https://onlinelibrary.wiley.com/doi/abs/10.1111/biom.13024},
}

@article{jiang_fence_2008,
  author =        {Jiang, Jiming and Rao, J. Sunil and Gu, Zhonghua and
                   Nguyen, Thuan},
  journal =       {The Annals of Statistics},
  month =         aug,
  note =          {Publisher: Institute of Mathematical Statistics},
  number =        {4},
  pages =         {1669--1692},
  title =         {Fence methods for mixed model selection},
  volume =        {36},
  year =          {2008},
  doi =           {10.1214/07-AOS517},
  issn =          {0090-5364, 2168-8966},
  url =           {https://projecteuclid.org/journals/annals-of-statistics/
                  volume-36/issue-4/Fence-methods-for-mixed-model-selection/
                  10.1214/07-AOS517.full},
}

@book{bickel_mathematical_2015,
  author =        {Bickel, Peter .J. and Doksum, Kjell A.},
  edition =       {0},
  month =         dec,
  publisher =     {Chapman and Hall/CRC},
  title =         {Mathematical {Statistics}: {Basic} {Ideas} and
                   {Selected} {Topics}, {Volumes} {I}-{II} {Package}},
  year =          {2015},
  doi =           {10.1201/9781315369266},
  isbn =          {978-1-315-36926-6},
  language =      {en},
  url =           {https://www.taylorfrancis.com/books/9781498740418},
}

@article{ferrari_confidence_2015,
  author =        {Ferrari, Davide and Yang, Yuhong},
  journal =       {Statistica Sinica},
  title =         {Confidence sets for model selection by {F} -testing},
  year =          {2015},
  doi =           {10.5705/ss.2014.110},
  issn =          {10170405},
  language =      {en},
  url =           {http://www3.stat.sinica.edu.tw/statistica/J25N4/J25N418/
                  J25N418.html},
}

@article{hansen_model_2011,
  author =        {Hansen, Peter R. and Lunde, Asger and
                   Nason, James M.},
  journal =       {Econometrica},
  note =          {Publisher: [Wiley, Econometric Society]},
  number =        {2},
  pages =         {453--497},
  title =         {The {Model} {Confidence} {Set}},
  volume =        {79},
  year =          {2011},
  issn =          {0012-9682},
  url =           {https://www.jstor.org/stable/41057463},
}

@misc{zhang_winners_2025,
  author =        {Zhang, Tianyu and Lee, Hao and Lei, Jing},
  month =         jul,
  note =          {arXiv:2408.02060 [math]},
  publisher =     {arXiv},
  title =         {Winners with {Confidence}: {Discrete} {Argmin}
                   {Inference} with an {Application} to {Model}
                   {Selection}},
  year =          {2025},
  doi =           {10.48550/arXiv.2408.02060},
  url =           {http://arxiv.org/abs/2408.02060},
}

@article{valiant_theory_1984,
  author =        {Valiant, L. G.},
  journal =       {Communications of the ACM},
  month =         nov,
  number =        {11},
  pages =         {1134--1142},
  title =         {A theory of the learnable},
  volume =        {27},
  year =          {1984},
  doi =           {10.1145/1968.1972},
  issn =          {0001-0782, 1557-7317},
  language =      {en},
  url =           {https://dl.acm.org/doi/10.1145/1968.1972},
}

@book{mohri_foundations_2018,
  address =       {Cambridge, MA},
  author =        {Mohri, Mehryar and Rostamizadeh, Afshin and
                   Talwalkar, Ameet},
  edition =       {Second edition},
  publisher =     {MIT Press},
  series =        {Adaptive computation and machine learning series},
  title =         {Foundations of machine learning},
  year =          {2018},
  isbn =          {978-0-262-03940-6 978-0-262-35136-2},
  language =      {eng},
}

@inproceedings{feldman_generalization_2018,
  author =        {Feldman, Vitaly and Vondrak, Jan},
  booktitle =     {Advances in {Neural} {Information} {Processing}
                   {Systems}},
  publisher =     {Curran Associates, Inc.},
  title =         {Generalization {Bounds} for {Uniformly} {Stable}
                   {Algorithms}},
  volume =        {31},
  year =          {2018},
  url =           {https://papers.nips.cc/paper_files/paper/2018/hash/
                  05a624166c8eb8273b8464e8d9cb5bd9-Abstract.html},
}

@inproceedings{dziugaite_computing_2017,
  author =        {Dziugaite, Gintare K. and Roy, Daniel M.},
  booktitle =     {Proceedings of the {Conference} on {Uncertainty} in
                   {Artificial} {Intelligence}},
  title =         {Computing nonvacuous generalization bounds for deep
                   (stochastic) neural networks with many more
                   parameters than training data.},
  year =          {2017},
}

@article{friedman_multivariate_1991,
  author =        {Friedman, Jerome H.},
  journal =       {The Annals of Statistics},
  month =         mar,
  note =          {Publisher: Institute of Mathematical Statistics},
  number =        {1},
  pages =         {1--67},
  title =         {Multivariate {Adaptive} {Regression} {Splines}},
  volume =        {19},
  year =          {1991},
  doi =           {10.1214/aos/1176347963},
  issn =          {0090-5364, 2168-8966},
  url =           {https://projecteuclid.org/journals/annals-of-statistics/
                  volume-19/issue-1/Multivariate-Adaptive-Regression-Splines/
                  10.1214/aos/1176347963.full},
}

@article{mentch_randomization_2020,
  author =        {Mentch, Lucas and Zhou, Siyu},
  journal =       {Journal of Machine Learning Research},
  number =        {171},
  pages =         {1--36},
  title =         {Randomization as {Regularization}: {A} {Degrees} of
                   {Freedom} {Explanation} for {Random} {Forest}
                   {Success}},
  volume =        {21},
  year =          {2020},
  issn =          {1533-7928},
  url =           {http://jmlr.org/papers/v21/19-905.html},
}

@article{hastie_best_2020,
  author =        {Hastie, Trevor and Tibshirani, Robert and
                   Tibshirani, Ryan},
  journal =       {Statistical Science},
  month =         nov,
  note =          {Publisher: Institute of Mathematical Statistics},
  number =        {4},
  pages =         {579--592},
  title =         {Best {Subset}, {Forward} {Stepwise} or {Lasso}?
                   {Analysis} and {Recommendations} {Based} on
                   {Extensive} {Comparisons}},
  volume =        {35},
  year =          {2020},
  doi =           {10.1214/19-STS733},
  issn =          {0883-4237, 2168-8745},
  url =           {https://projecteuclid.org/journals/statistical-science/
                  volume-35/issue-4/Best-Subset-Forward-Stepwise-or-Lasso-
                  Analysis-and-Recommendations-Based/10.1214/19-STS733.full},
}

@misc{barry_becker_adult_1996,
  author =        {Barry Becker, Ronny Kohavi},
  publisher =     {UCI Machine Learning Repository},
  title =         {Adult},
  year =          {1996},
  doi =           {10.24432/C5XW20},
  url =           {https://archive.ics.uci.edu/dataset/2},
}

@misc{hofmann_statlog_1994,
  author =        {Hofmann, Hans},
  publisher =     {UCI Machine Learning Repository},
  title =         {Statlog ({German} {Credit} {Data})},
  year =          {1994},
  doi =           {10.24432/C5NC77},
  url =           {https://archive.ics.uci.edu/dataset/144},
}

@misc{larson_how_2016,
  author =        {Larson, Jeff and Surya, Mattu and Kirchner, Lauren and
                   Angwin, Julia},
  journal =       {ProPublica},
  month =         may,
  title =         {How {We} {Analyzed} the {COMPAS} {Recidivism}
                   {Algorithm}},
  year =          {2016},
  language =      {en},
  url =           {https://www.propublica.org/article/how-we-analyzed-the-
                  compas-recidivism-algorithm},
}

@misc{cognos_analytics_telco_2018,
  author =        {Cognos Analytics, IBM},
  publisher =     {IBM Sample Data Sets},
  title =         {Telco {Customer} {Churn}},
  year =          {2018},
}

@misc{wnek_monks_1993,
  author =        {Wnek, J.},
  publisher =     {UCI Machine Learning Repository},
  title =         {{MONK}'s {Problems}},
  year =          {1993},
  doi =           {https://doi.org/10.24432/C5R30R},
}

@article{harrison_hedonic_1978,
  author =        {Harrison, David and Rubinfeld, Daniel L},
  journal =       {Journal of Environmental Economics and Management},
  month =         mar,
  number =        {1},
  pages =         {81--102},
  title =         {Hedonic housing prices and the demand for clean air},
  volume =        {5},
  year =          {1978},
  doi =           {10.1016/0095-0696(78)90006-2},
  issn =          {0095-0696},
  url =           {https://www.sciencedirect.com/science/article/pii/
                  0095069678900062},
}

@misc{quinlan_auto_1993,
  author =        {Quinlan, R},
  publisher =     {UCI Machine Learning Repository},
  title =         {Auto {MPG}},
  year =          {1993},
  doi =           {https://doi.org/10.24432/C5859H},
}

@article{fanaee-t_event_2014,
  author =        {Fanaee-T, Hadi and Gama, Joao},
  journal =       {Progress in Artificial Intelligence},
  month =         jun,
  number =        {2},
  pages =         {113--127},
  title =         {Event labeling combining ensemble detectors and
                   background knowledge},
  volume =        {2},
  year =          {2014},
  doi =           {10.1007/s13748-013-0040-3},
  issn =          {2192-6360},
  language =      {en},
  url =           {https://doi.org/10.1007/s13748-013-0040-3},
}

@phdthesis{waugh_extending_1995,
  author =        {Waugh, S. G.},
  month =         jan,
  school =        {University of Tasmania},
  type =          {thesis},
  title =         {Extending and benchmarking {Cascade}-{Correlation} :
                   extensions to the {Cascade}-{Correlation}
                   architecture and benchmarking of feed-forward
                   supervised artificial neural networks},
  year =          {1995},
  doi =           {10.25959/23242847.v1},
  language =      {en},
  url =           {https://figshare.utas.edu.au/articles/thesis/
                  Extending_and_benchmarking_Cascade-
                  Correlation_extensions_to_the_Cascade-
                  Correlation_architecture_and_benchmarking_of_feed-
                  forward_supervised_artificial_neural_networks/23242847/1},
}

@article{yeh_modeling_1998,
  author =        {Yeh, I. -C.},
  journal =       {Cement and Concrete Research},
  month =         dec,
  number =        {12},
  pages =         {1797--1808},
  title =         {Modeling of strength of high-performance concrete
                   using artificial neural networks},
  volume =        {28},
  year =          {1998},
  doi =           {10.1016/S0008-8846(98)00165-3},
  issn =          {0008-8846},
  url =           {https://www.sciencedirect.com/science/article/pii/
                  S0008884698001653},
}

@article{ein-dor_attributes_1987,
  author =        {Ein-Dor, Phillip and Feldmesser, Jacob},
  journal =       {Commun. ACM},
  month =         apr,
  number =        {4},
  pages =         {308--317},
  title =         {Attributes of the performance of central processing
                   units: a relative performance prediction model},
  volume =        {30},
  year =          {1987},
  doi =           {10.1145/32232.32234},
  issn =          {0001-0782},
  url =           {https://dl.acm.org/doi/10.1145/32232.32234},
}

@inproceedings{ahmed_using_2015,
  author =        {Ahmed, Mehreen and Jahangir, Maham and Afzal, Hammad and
                   Majeed, Awais and Siddiqi, Imran},
  booktitle =     {2015 {IEEE} {International} {Conference} on {Smart}
                   {City}/{SocialCom}/{SustainCom} ({SmartCity})},
  month =         dec,
  pages =         {273--278},
  title =         {Using {Crowd}-{Source} {Based} {Features} from
                   {Social} {Media} and {Conventional} {Features} to
                   {Predict} the {Movies} {Popularity}},
  year =          {2015},
  doi =           {10.1109/SmartCity.2015.83},
  url =           {https://ieeexplore.ieee.org/document/7463737},
}

@article{moro_predicting_2016,
  author =        {Moro, Sérgio and Rita, Paulo and Vala, Bernardo},
  journal =       {Journal of Business Research},
  month =         sep,
  number =        {9},
  pages =         {3341--3351},
  title =         {Predicting social media performance metrics and
                   evaluation of the impact on brand building: {A} data
                   mining approach},
  volume =        {69},
  year =          {2016},
  doi =           {10.1016/j.jbusres.2016.02.010},
  issn =          {0148-2963},
  url =           {https://www.sciencedirect.com/science/article/pii/
                  S0148296316000813},
}

@misc{athanasios_tsanas_parkinsons_2009,
  author =        {Athanasios Tsanas, Max Little},
  publisher =     {UCI Machine Learning Repository},
  title =         {Parkinsons {Telemonitoring}},
  year =          {2009},
  doi =           {10.24432/C5ZS3N},
  url =           {https://archive.ics.uci.edu/dataset/189},
}

@inproceedings{quinlan_combining_1993,
  address =       {San Francisco, CA, USA},
  author =        {Quinlan, J. Ross},
  booktitle =     {Proceedings of the {Tenth} {International}
                   {Conference} on {International} {Conference} on
                   {Machine} {Learning}},
  month =         jul,
  pages =         {236--243},
  publisher =     {Morgan Kaufmann Publishers Inc.},
  series =        {{ICML}'93},
  title =         {Combining instance-based and model-based learning},
  year =          {1993},
  isbn =          {978-1-55860-307-3},
}

@article{bergstrom_molecular_2003,
  author =        {Bergström, Christel A. S. and Norinder, Ulf and
                   Luthman, Kristina and Artursson, Per},
  journal =       {Journal of Chemical Information and Computer
                   Sciences},
  month =         jul,
  note =          {Publisher: American Chemical Society},
  number =        {4},
  pages =         {1177--1185},
  title =         {Molecular {Descriptors} {Influencing} {Melting}
                   {Point} and {Their} {Role} in {Classification} of
                   {Solid} {Drugs}},
  volume =        {43},
  year =          {2003},
  doi =           {10.1021/ci020280x},
  issn =          {0095-2338},
  url =           {https://doi.org/10.1021/ci020280x},
}

@book{lehmann_testing_2010,
  address =       {New York, NY},
  author =        {Lehmann, Erich L. and Romano, Joseph P.},
  edition =       {3., rd ed. 2005. Corr. 2nd printing. Softcover
                   version of original hardcover edition 2005},
  publisher =     {Springer New York},
  series =        {Springer texts in statistics},
  title =         {Testing statistical hypotheses},
  year =          {2010},
  isbn =          {978-1-4419-3178-8},
  language =      {eng},
}

@article{wilks_large-sample_1938,
  author =        {Wilks, S. S.},
  journal =       {The Annals of Mathematical Statistics},
  month =         mar,
  note =          {Publisher: Institute of Mathematical Statistics},
  number =        {1},
  pages =         {60--62},
  title =         {The {Large}-{Sample} {Distribution} of the
                   {Likelihood} {Ratio} for {Testing} {Composite}
                   {Hypotheses}},
  volume =        {9},
  year =          {1938},
  doi =           {10.1214/aoms/1177732360},
  issn =          {0003-4851, 2168-8990},
  url =           {https://projecteuclid.org/journals/annals-of-mathematical-
                  statistics/volume-9/issue-1/The-Large-Sample-Distribution-of-
                  the-Likelihood-Ratio-for-Testing/10.1214/aoms/
                  1177732360.full},
}

@article{neyman_ix_1997,
  author =        {Neyman, Jerzy and Pearson, Egon Sharpe and
                   Pearson, Karl},
  journal =       {Philosophical Transactions of the Royal Society of
                   London. Series A, Containing Papers of a Mathematical
                   or Physical Character},
  month =         jan,
  note =          {Publisher: Royal Society},
  number =        {694-706},
  pages =         {289--337},
  title =         {{IX}. {On} the problem of the most efficient tests of
                   statistical hypotheses},
  volume =        {231},
  year =          {1997},
  doi =           {10.1098/rsta.1933.0009},
  url =           {https://royalsocietypublishing.org/doi/10.1098/
                  rsta.1933.0009},
}

@book{boucheron_concentration_2013,
  author =        {Boucheron, Stéphane and Lugosi, Gábor and
                   Massart, Pascal},
  month =         feb,
  publisher =     {Oxford University Press},
  title =         {Concentration {Inequalities}: {A} {Nonasymptotic}
                   {Theory} of {Independence}},
  year =          {2013},
  doi =           {10.1093/acprof:oso/9780199535255.001.0001},
  isbn =          {978-0-19-174710-6},
  language =      {en},
  url =           {https://academic.oup.com/book/26549},
}

@book{chafai_interactions_2012,
  address =       {Paris},
  author =        {Chafaï, Djalil and Guédon, Olivier and
                   Lecué, Guillaume and Pajor, Alain},
  note =          {OCLC: ocn861119770},
  number =        {no. 37},
  publisher =     {Société Mathématique de France},
  series =        {Panoramas et synthèses},
  title =         {Interactions between compressed sensing random
                   matrices and high dimensional geometry},
  year =          {2012},
  isbn =          {978-2-85629-370-6},
  language =      {eng},
}

@article{hara_approximate_2018,
  author =        {Hara, Satoshi and Ishihata, Masakazu},
  journal =       {Proceedings of the AAAI Conference on Artificial
                   Intelligence},
  month =         apr,
  note =          {Number: 1},
  number =        {1},
  title =         {Approximate and {Exact} {Enumeration} of {Rule}
                   {Models}},
  volume =        {32},
  year =          {2018},
  doi =           {10.1609/aaai.v32i1.11637},
  issn =          {2374-3468},
  language =      {en},
  url =           {https://ojs.aaai.org/index.php/AAAI/article/view/11637},
}

@misc{mata_computing_2022,
  author =        {Mata, Kota and Kanamori, Kentaro and Arimura, Hiroki},
  month =         apr,
  note =          {arXiv:2204.11285 [cs]},
  publisher =     {arXiv},
  title =         {Computing the {Collection} of {Good} {Models} for
                   {Rule} {Lists}},
  year =          {2022},
  doi =           {10.48550/arXiv.2204.11285},
  url =           {http://arxiv.org/abs/2204.11285},
}

@article{xin_exploring_2022,
  author =        {Xin, Rui and Zhong, Chudi and Chen, Zhi and
                   Takagi, Takuya and Seltzer, Margo and Rudin, Cynthia},
  journal =       {Advances in Neural Information Processing Systems},
  month =         dec,
  pages =         {14071--14084},
  title =         {Exploring the {Whole} {Rashomon} {Set} of {Sparse}
                   {Decision} {Trees}},
  volume =        {35},
  year =          {2022},
  language =      {en},
  url =           {https://proceedings.neurips.cc/paper_files/paper/2022/hash/
                  5afaa8b4dd18eb1eed055d2d821b58ae-Abstract-Conference.html},
}

@article{laberge_partial_2023,
  author =        {Laberge, Gabriel and Pequignot, Yann and
                   Mathieu, Alexandre and Khomh, Foutse and
                   Marchand, Mario},
  journal =       {Journal of Machine Learning Research},
  number =        {364},
  pages =         {1--50},
  title =         {Partial {Order} in {Chaos}: {Consensus} on {Feature}
                   {Attributions} in the {Rashomon} {Set}},
  volume =        {24},
  year =          {2023},
  issn =          {1533-7928},
  url =           {http://jmlr.org/papers/v24/23-0149.html},
}

@article{zhong_exploring_2023,
  author =        {Zhong, Chudi and Chen, Zhi and Liu, Jiachang and
                   Seltzer, Margo and Rudin, Cynthia},
  journal =       {Advances in Neural Information Processing Systems},
  month =         dec,
  pages =         {56673--56699},
  title =         {Exploring and {Interacting} with the {Set} of {Good}
                   {Sparse} {Generalized} {Additive} {Models}},
  volume =        {36},
  year =          {2023},
  language =      {en},
  url =           {https://proceedings.neurips.cc/paper_files/paper/2023/hash/
                  b1719f44953c2e0754a016ab267fe4e7-Abstract-Conference.html},
}

@article{dey_multiple_2026,
  author =        {Dey, Neil and Martin, Ryan and Williams, Jonathan P.},
  journal =       {Statistics \& Probability Letters},
  month =         feb,
  pages =         {110559},
  title =         {Multiple testing in generalized universal inference},
  volume =        {228},
  year =          {2026},
  doi =           {10.1016/j.spl.2025.110559},
  issn =          {0167-7152},
  url =           {https://www.sciencedirect.com/science/article/pii/
                  S0167715225002044},
}
\newpage

\appendix

\section{Proofs of Main Results}

\subsection{Proof of Proposition \ref{prop:uniform}}

\begin{proof}
  We provide the proof steps without considering the parenthetical limit \( (as) \) because whether \( (as) \) is included or not, the steps are identical. First, assume for some \( r \in [d] \) that  (\ref{eq:mss-test-condition}) and (\ref{eq:rashomon-set}) are satisfied. Then, by definition \[ \sup_{P \in \mathcal{P}_{0,r}} P \left \{ \psi_r = 1 \right \}  = \sup_{P \in \mathcal{P}_{0,r}} P \left \{ \inf_{t \in T_r} \phi \left ( h_t \right )  = 1 \right \} \le  \sup_{P \in \mathcal{P}_{0,r}} P \left \{  \phi \left ( f_r \right ) = 1 \right \}.\]
  Assuming \( P \in \mathcal{P}_{0,r} \), by the assumption of (\ref{eq:rashomon-set}), \( R_n \left ( f_{r} \right ) \le \argmin_{k \in [d]} \inf_{t \in T_k} R \left ( h_t \right ) + \epsilon \) which implies \( r \in \Theta^{MSS}_{\epsilon, n} \). Thus,
  \[   \sup_{P \in \mathcal{P}_{0,r}} P \left \{  \phi \left ( f_r \right ) = 1 \right \} = \sup_{P \in \mathcal{P}_{0,r}} P \left \{  \phi \left ( f_r \right ) = 1 \cap r \in \Theta^{MSS}_{n,\epsilon} \right \} \le \alpha \]
  which completes the first result. Next, assume for all \( r \in [d] \) that (\ref{eq:rashomon-set}) and (\ref{eq:mss-test-condition}) are satisfied.  By the previous result,
  \[   \inf_{P \in \mathcal{P}} \inf_{r \in \Theta \left ( P \right )} P \left \{  r \in \hat{\Theta} \right \}  \ge 1 - \sup_{P \in \mathcal{P}} \sup_{r \in \Theta \left ( P \right )} P \left \{  \psi_r  = 1 \right \} \ge 1 - \alpha.\]

\end{proof}

\subsection{Proof of Proposition \ref{prop:selective}}

\begin{proof}

  First, assume for some \( r \in [d] \) that  (\ref{eq:mss-test-condition}) and (\ref{cond:rashomon-convergence}) are satisfied. Let \( E_{n,r} \) be the event that \(R_n \left ( f_r \right ) \le \inf_{t \in T_r} R \left ( h_t \right ) + \epsilon\). Then, by definition
  \begin{align*}  \limsup_{n \to \infty} \sup_{P \in \mathcal{P}_{0,r}} P \left \{  \psi_r = 1 \right \} &=  \limsup_{n \to \infty} \sup_{P \in \mathcal{P}_{0,r}} P \left \{  \phi \left ( f_r \right ) = 1 \right \} \\
  &=  \limsup_{n \to \infty} \sup_{P \in \mathcal{P}_{0,r}} \left [ P \left \{ \phi \left ( f_r \right ) = 1 \cap E_{n,r} \right \} + P \left \{ \phi \left ( f_r \right ) = 1 \cap E^{C}_{n,r} \right \}\right ].  \end{align*}
By (\ref{cond:rashomon-convergence}), \( P \left \{ E_{n,r}^C \right \}  \rightarrow 0 \) as \( n \rightarrow \infty \). This implies that \( \sup_{P \in \mathcal{P}_{0,r}} P \left \{ \phi \left ( f_r \right ) = 1 \cap E^{C}_{n,r} \right \} \rightarrow 0 \) as \( n \rightarrow \infty \). Furthermore, assuming \( P \in \mathcal{P}_{0,r} \) and \( E_{n,r} \) is true, then \( R_n \left ( f_{r} \right ) \le \argmin_{k \in [d]} \inf_{t \in T_k} R \left ( h_t \right ) + \epsilon \). Thus, in such a case, \( r \in \Theta^{MSS}_{n,\epsilon} \). This implies that
\begin{align*}  \limsup_{n \to \infty} \sup_{P \in \mathcal{P}_{0,r}}  P \left \{ \phi \left ( f_r \right ) = 1 \cap E_{n,r} \right \} &= \limsup_{n \to \infty} \sup_{P \in \mathcal{P}_{0,r}}  P \left \{ \phi \left ( f_r \right ) = 1 \cap E_{n,r} \cap r \in \Theta^{MSS}_{n,\epsilon} \right \} \\
&\le \limsup_{n \to \infty} \sup_{P \in \mathcal{P}_{0,r}}  P \left \{ \phi \left ( f_r \right ) = 1 \cap r \in \Theta^{MSS}_{n,\epsilon}\right \} \\
&\le \alpha\end{align*}
which completes the first result. The point-wise coverage result follows in the same manner as described in the Proof of Proposition \ref{prop:uniform}.

\end{proof}

  \subsection{Proof of Theorem \ref{thm:sui-holdout}}

  \begin{proof}
     Assume that \( \left ( \nabla_{n,1}(f_r, f_{s_r}, Z_{0})  - \mu_n \right )^{2} / \sigma_n^2  \) is uniformly integrable. Without loss of generality, assume that \( \ndata_{-1} = \left \{ Z_1, \dots, Z_{n_{te}} \right \} \). For ease of notation, let \( X_{n,i} = \nabla_{n,1} \left ( f_r, f_{s_r}, Z_i \right ) \). Then, note that since \( k_n = 1 \),
  \[ \frac{(n_{te} - 1) \bar{\sigma}^2_n \left ( f_r, f_{s_r} \right )}{n_{te}}  = \frac{1}{n_{te}} \sum_{i = 1}^{n_{te}} \left [ \left ( X_{n,i} - \mu_n \right )^{2} - ( \bar{R}_n(f_r, f_{s_r} ) - \mu_n)^{2} \right ].   \]
 Let \( Y_{n,i} =  \frac{ X_{n, i} - \mu_n }{\sigma_n} \). By the uniform integrability assumption,
  \[  \lim_{\lambda \to \infty} \limsup_{n \to \infty} \mathbb{E} \left [ Y^{2}_{n,i} I \left \{ Y^{2}_{n,i} > \lambda \right \} \right ] = 0,  \]
 which implies 
 \[  \lim_{\lambda \to \infty} \limsup_{n \to \infty} \mathbb{E} \left [ |Y_{n,i}| I \left \{ |Y_{n,i}| > \lambda \right \} \right ] = 0.  \]
 Thus, by \citep[Lemma 11.4.2]{lehmann_testing_2010},
 \[ \frac{\bar{R}_n(f_r, f_{s_r}) - \mu_n}{\sigma_n} = \frac{1}{n_{te}} \sum_{i = 1}^{n_{te}} Y_{n,i} \overset{p}{\to} 0. \]
By Lemma 11.4.3 of the same work,
 \[\frac{ \frac{1}{n_{te}} \sum_{i = 1}^{n_{te}} \left (X_{n,i} - \mu_{n} \right )^{2}}{\sigma_n^{2}} =   \frac{1}{n_{te}} \sum_{i = 1}^{n_{te}} Y^{2}_{n,i} \overset{p}{\to} 1.  \]
 It then follows by the continuous mapping theorem that
 \[ \frac{\bar{\sigma}_n(f_r, f_{s_r})}{\sigma_{n}} \overset{p}{\rightarrow} 1.  \]
 Furthermore,
 \[ \frac{ \sqrt{n_{te}} \left [ \bar{R}_n(f_r, f_{s_r}) - \mu_n \right ]}{ \sigma_n  } \overset{d}{\to} N(0,1)   \]
 by \citep[Lemma 11.4.1]{lehmann_testing_2010}. Note that when \( r \in \Theta^{MSS}_{\epsilon,n} \), \( \mu_n - \epsilon \le 0 \). Thus, for any \( P \in \mathcal{P} \), by Slutsky's Lemma and Polya's Theorem,
 \begin{align*}
   \limsup_{n \rightarrow \infty} P \{  \phi_{CLT}  \left ( f_r \right ) = 1  \;\cap\; &E_{r, \epsilon, n}  \} = \limsup_{n \to \infty} P \left \{   \frac{\sqrt{n_{te}}}{\bar{\sigma}_n \left ( f_r, f_{s_r} \right )} \bar{R}_n \left ( f_r, f_{s_r} \right ) > \delta \left ( \alpha \right ) + \epsilon \cap  E_{r, \epsilon, n} \right \} \\
                                                                   &\le \limsup_{n \to \infty} P \left \{   \frac{\sqrt{n_{te}}}{\bar{\sigma}_n \left ( f_r, f_{s_r} \right )} \left (\bar{R}_n \left ( f_r, f_{s_r} \right ) - \mu_n \right ) > \delta \left ( \alpha \right ) \cap  E_{r, \epsilon, n} \right \} \\
                                                                   &\le \limsup_{n \to \infty} P \left \{   \frac{\sqrt{n_{te}}}{\bar{\sigma}_n \left ( f_r, f_{s_r} \right )} \left (\bar{R}_n \left ( f_r, f_{s_r} \right ) - \mu_n \right ) > \delta \left ( \alpha \right ) \right \} \\
    &= \limsup_{n \to \infty} P \left \{   \frac{\sqrt{n_{te}}}{\sigma_n \left ( f_r, f_{s_r} \right )} \left (\bar{R}_n \left ( f_r, f_{s_r} \right ) - \mu_n \right ) > \delta \left ( \alpha \right ) \right \} \\
        &\le \alpha
 \end{align*}
 where \( \delta \left ( \alpha \right ) = \Phi^{-1} ( 1- \alpha) \) and \( E_{r, \epsilon, n} \) denotes the event that \( r \in \Theta^{MSS}_{\epsilon, n} \).
  \end{proof}

\subsection{Proof of Theorem \ref{thm:sui-cv}}

\begin{proof}
Assume the sequence \( \left ( \bar{\nabla}_n(Z_0) - \bar{\mu}_n \right )^{2} / \sigma_n^2  \) is uniformly integrable and \( \gamma_{n_{tr}}^{loss}(g)  = o \left ( \sigma^2_n / n \right ) \). Let \( \mu_n = R_n \left ( f_r \right ) - R_n \left ( f_{s_r} \right ) \). By \citep[Theorems 2 and 4]{bayle_cross-validation_2020},
\[ \frac{\bar{\sigma}_n \left ( f_r, f_{s_r} \right )}{\sigma_n} \overset{p}{\to} 1 \text{ and }  \frac{\sqrt{n}}{\sigma_n} \left ( \bar{R}_n \left ( f_r, f_{s_r} \right ) - \mu_{n} \right ) \overset{d}{\to} N(0,1).\]
Therefore, by the same steps as the end of the proof of Theorem \ref{thm:sui-holdout}, \( \limsup_{n \rightarrow \infty} P \{  \phi_{CLT} \left ( f_r \right ) = 1  \cap r \in \Theta^{MSS}_{n,\epsilon} \}  \le \alpha \) for any \( P \in \mathcal{P} \).
\end{proof}

  \subsection{Proof of Theorem \ref{thm:ui}}

  \begin{proof}
    By the law of total expectation and the fact that each \( Z \in \ndata \) are i.i.d.,
    \begin{align*}
     \mathbb{E} \left [ \bar{R}_n^{\exp} \left ( f_r, f_{s_r}, \omega \right ) \right ] &= k_n^{-1} \sum_{j=1}^{k_n} \mathbb{E} \left [ \exp \left \{  \omega \sum_{Z \in \ndata_{-j}} \nabla_{n,j} \left ( f_r, f_{s_r}, Z \right ) \right \}  \right ] \\
      &= k_n^{-1} \sum_{j=1}^{k_n}  \mathbb{E} \left [ \prod_{Z \in \ndata_{-j}} \mathbb{E}_{Z} \left [ \exp \left \{  \omega \nabla_{n,j} \left ( f_r, f_{s_r}, Z \right ) \right \}  \right ] \right ].
   \end{align*}
   By the strong central condition, \( \mathbb{E}_{Z} \left [ \exp \left \{  \omega \nabla_{n,j} \left ( f_r, f_{s_r}, Z \right ) \right \} \right ] \le 1 \) almost surely for any \( j \in [k_n] \) and \( Z \in \ndata_{-j} \). Thus, \( \mathbb{E} \left [ \bar{R}_n^{\exp} \left ( f_r, f_{s_r}, \omega \right ) \right ] \le 1 \). By Markov's Inequality, for any \( P \in \mathcal{P} \),
\begin{align*}
  P \left ( \phi_{UI} \left ( f_r \right ) =  1 \cap  r \in \Theta^{MSS}_{\epsilon, n} \right ) &\le P \left (\bar{R}_n^{\exp} \left ( f_r, f_{s_r}, \omega \right ) > \alpha^{-1}  \cap  r \in \Theta^{MSS}_{\epsilon, n} \right ) \\
  &\le P \left (\bar{R}_n^{\exp} \left ( f_r, f_{s_r}, \omega \right ) > \alpha^{-1} \right ) \\
  &\le \alpha \mathbb{E} \left [  \bar{R}_n^{\exp} \left ( f_r, f_{s_r}, \omega \right )  \right ] \\
      &\le \alpha.
\end{align*}
\end{proof}
  \section{Further Discussion and Results}
  
\subsection{A Conservative Version of \( \phi_{CLT} \)}
\label{sec:conservative-studentization}

\begin{figure}[hbt!]
 \center
 \includegraphics[width = \textwidth]{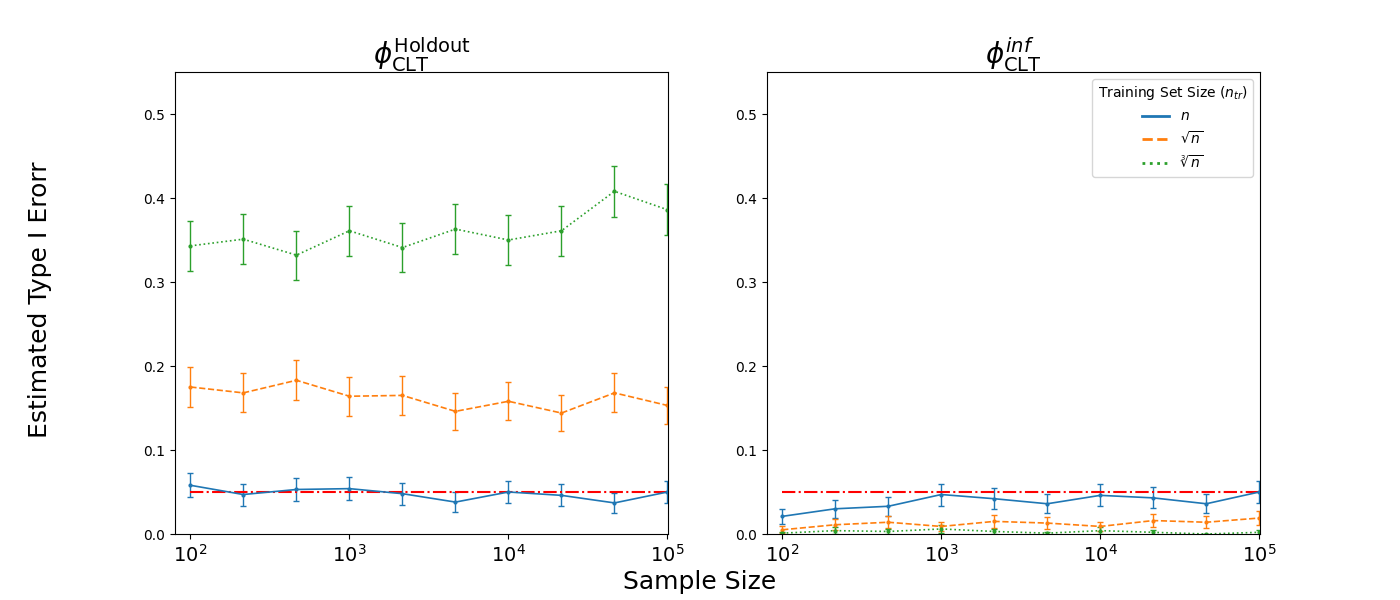}
 \caption{Example of a simulated regression scenario where \( \phi_{CLT}^{Holdout} \) fails to control the type I error while \( \phi^{inf}_{CLT} \) does. For \( \phi^{inf}_{CLT} \), we set \( k_n = 1 \) and use the same holdout split as \( \phi_{CLT}^{Holdout} \). Let \( Z_i = (Y_i, X_i) \in \mathbb{R} \times \mathbb{R}^3\) where \( Y_i = X_i \beta + \epsilon_i\), \( X_i \iid N(0_{3}, I_3)\), \( \beta = (1, 1, 1)^{\top} \), and \( \epsilon_i \iid N(0,1) \). \( \mc_{1} = \left \{ h_{t} \; : \; t \in \mathbb{R}^3, t_{2} = 0 \right \}  \) and \( \mc_{2} = \left \{ h_{t} \; : \; t \in \mathbb{R}^3, t_{3} = 0 \right \}  \) correspond to classes of additive linear models where the second and third features are not used, respectively. The loss, \( \ell \), is the squared error loss: \( \ell \left (h_{\theta}, Z_{i} \right ) = \left ( X_{i} \theta - Y_i \right )^2 \). The left and right plots correspond to using the tests \( \phi_{CLT}^{Holdout} \) and \( \phi_{CLT}^{\inf} \) to test \( H_{0,1}: 1 \in \Theta \). Note by definition that \( H_{0,1} \) is true. We set \( f_{1} \) and \( f_2 \) to output the OLS esimates for \( \mc_1 \) and \( \mc_2 \), \( s_1 = 2 \), and \( \epsilon = 1\times10^{-6} \). The solid, dashed, and dotted lines correspond to the estimated type I error across 1,000 simulations for training set sizes \( n/2 \), \( \sqrt{n} \), and \( \sqrt[3]{n} \), respectively. The dashed/dotted horizontal line depicts the significance level \( \alpha = .05 \).}
 \label{fig:sui-convergence-example}
\end{figure}

In Section \ref{sec:studentization}, the studentized test \( \phi_{CLT} \) was proposed. Recall that for \( \phi_{CLT} \) to be combined with Proposition \ref{prop:selective} to conduct valid MCS, we require condition (\ref{cond:rashomon-convergence}). There, however, exist some settings in which  \( R_n \left (f_r  \right )- \inf_{t \in T_r} R \left ( h_t \right ) - \epsilon \le 0 \) is only likely to occur in extremely large sample sizes. A simple example of such an atypical setting is when the training set size \( n_{tr} \) is much smaller than the test set size \( n_{te} \). Figure \ref{fig:sui-convergence-example} showcases a simulated regression scenario where the MCS test \( \phi^{Holdout}_{CLT} \) as defined in Section \ref{sec:experiments}, fails to appropriately control the type I error even when the sample size is very large.

Define the test
\begin{equation}\label{eq:suiinf}  \phi^{inf}_{CLT} \left ( f_r\right ) = I \left \{ \inf_{h \in \mc_{r}} \bar{R}_n \left (f_{h}, f_{s_r} \right ) > k_n^{-1/2} n_{te}^{-1/2} \bar{\sigma}_n \left ( f_r, f_{s_r} \right ) \left ( \Phi^{-1} \left ( 1 - \alpha \right ) + \epsilon \right )  \right \} .  \end{equation}
Note that \( \inf_{h \in \mc_r} \bar{R}_n \left (f_{h}, f_{s_r} \right ) \) is a conservative estimate for \( \bar{R}_n \left (f_{r}, f_{s_r} \right ) \) when \( k_n = 1 \). The next result shows that \( \phi_{CLT}^{inf} \) satisfies similar asymptotic properties to \( \phi_{CLT} \). 

\begin{corollary}[Validity of \( \phi^{inf}_{CLT} \) when \( k_n = 1 \)]
  \label{cor:suinf}
  Under the assumptions of Theorem \ref{thm:sui-holdout}, the test (\ref{eq:suiinf}) 
satisfies the asymptotic version of (\ref{eq:mss-test-condition}).
\end{corollary}
\begin{proof}
   Note that \( \inf_{h \in H_r} \bar{R}_n \left (f_{h}, f_{s_r} \right ) \le  \bar{R}_n \left (f_{r}, f_{s_r} \right ) \). Thus, \[ \limsup_{n \to \infty}  P \left \{  \phi^{inf}_{CLT} \left ( f_r \right ) = 1 \cap r \in \Theta_{n,\epsilon}^{MSS} \right \}  \le \limsup_{n \to \infty}  P \left \{  \phi_{CLT} \left ( f_r \right ) = 1 \cap r \in \Theta_{n,\epsilon}^{MSS}\right \}.  \]
  The result follows from Theorem \ref{thm:sui-holdout}.
\end{proof}

In contrast to \( \phi_{CLT} \), the test \( \phi^{inf}_{CLT} \)  appropriately controls the type I error in the situation displayed by Figure \ref{fig:sui-convergence-example}.
A downside of \( \phi^{inf}_{CLT} \) is that \( \inf_{h \in \mc_r} \bar{R}_n \left (f_{h}, f_{s_r} \right)  \) has to be computable which limits its application to only model classes where an empirical risk minimizer is able to be found.

\subsection{Likelihood Ratio Test}
\label{sec:stat-models}

Assume for \( r \in [d] \) that \( \mc_r = \{P_{\theta}: \theta \in \Theta_{r} \} \) is a class of statistical models where each \( P_{\theta} \) for \( \theta \in \Theta_r \subseteq \Theta \subseteq \mathbb{R}^{p}\) denotes a unique distribution parameterized by the \( p \)-dimensional vector \( \theta \). Typically in such a setting, the goal is to correctly identify the parameter \( \theta \in \Theta \) which is most likely to have generated the data \( \left \{ Z_1, \dots, Z_n \right \}  \). With this in mind, let the loss function be the negative log-likelihood function
\[ \ell(P_{\theta}, z) = - \log \left (p_{\theta} \left ( z \right ) \right ) \]
where \( p_{\theta}(z) \) is the likelihood of observating \( z \in \mathcal{Z} \) under the distribution \( P_{\theta} \). Assuming the true data generating distribution \( P_{\theta^{*}} \) exists such that \( \theta^{*} \in \Theta \), then \( \theta^{*} \) will be the minimizer of \( R \left ( P_{\theta}\right ) \) across \( \theta \in \Theta \) \citep{bickel_mathematical_2015}.

Let
\[ L_n(\theta) = \sum_{i=1}^{n} \log \left (p_{\theta}\left (Z_{i} \right ) \right ) = \log \left ( \prod_{i=1}^n p_{\theta}(Z_i) \right )\]
be the empirical error for the parameter \( \theta \in \Theta \). Set \( \hat{\theta} = \arg \max_{\theta \in \Theta} \prod_{i=1}^n p_{\theta} \left ( Z_i \right ) \) to be the maximum likelihood estimate across \( \Theta \). Let \( \chi^2_p (1-\alpha) \) represent the upper \( \alpha \)-quantile of the chi-squared distribution with \( p \) degrees of freedom. Assume \( f_1, \dots, f_d \) represent empirical risk minimization strategies for \( \mc_1, \dots, \mc_d \). Then, the likelihood ratio test
\begin{equation}
  \label{eq:lrt-test}
  \phi \left ( f_{r} \right ) = I \left \{ -2 \left [ L_n \left ( \theta_{r} \right ) - L_n \left ( \hat{\theta} \right ) \right ] > \chi^{2}_{p}(1-\alpha)  \right \} \text{ where } P_{\theta_r} = f_{r} \left ( \ndata \right )
  \end{equation}
satisfies (\ref{eq:mss-test-condition}) as long as a few regularity conditions are satisfied \citep{wilks_large-sample_1938}. See \cite[Section 6.2]{bickel_mathematical_2015} or \cite[Section 3]{wasserman_universal_2020} for a discussion of the necessary regularity conditions which include statistical models being identifiable and differentiable in quadratic mean, a compact parameter space, and the log-likehood being a smooth function for any \( \theta \in \Theta \). In certain statistical settings the likelihood ratio test been shown to be the the uniformly most powerful test \citep{neyman_ix_1997}. Alternatives to the likelihood ratio statistic with similar forms of limiting distributions are the wald statistic and rao score test statistic \citep{bickel_mathematical_2015}.

We are not the first to notice the potential application of the asymptotics of the likelihood ratio test statistic for MCS. A likelihood ratio based methodology has been proposed to conduct valid MCS with collections of appropriately nested model classes \citep{zheng_model_2019,li_model_2019}. Another line of work proposes a similar method for model classes composed of linear mixed models \citep{jiang_fence_2008}. Furthermore, when \( \mc_1, \dots, \mc_d \) are composed of gaussian linear models, the likelihood ratio test is equivalent to conducting an F-test \citep[Proposition 6.1.2]{bickel_mathematical_2015} and provides exact coverage. Similar methodology relying on the F-test applicable to MCS for nested gaussian linear model classes has also been suggested \citep{ferrari_confidence_2015}. Lastly, it should be noted that the original UI methodology was inspired by the likelihood ratio test \citep{wasserman_universal_2020}.

\subsection{Concentration Inequality based Approaches}
\label{sec:ci}

For simplicity, let \( \ndata_0 \subseteq \ndata \) and assume in the notation of Section \ref{sec:proposed-methods} that \( k_n = 1 \). In the proof of Theorem \ref{thm:ui}, the strong central type condition, (\ref{eq:strong-central-condition}), was combined with Markov's inequality to construct a concentration inequality similar to
\begin{equation} \label{eq:ci} \mathbb{P} \left \{ \sum_{Z \in \ndata_0} \nabla_{n,1} \left ( f_r, f_{s_{r}}, Z \right) - \left [ R_n \left (f_r \right ) - R_n \left (f_{s_{r}} \right ) \right ] \le \delta_n(\alpha) \right \} \ge 1 - \alpha    \end{equation}
for some arbitrary choice of \( \delta_n \left ( \alpha \right ) \). A result like that of Theorem \ref{thm:ui} holds for this general inequality as well. Define the test
\begin{equation} \label{eq:concentration-inequality} \phi_{CI} \left ( f_r \right ) = I \left \{ \sum_{Z \in \ndata_0} \nabla_{n,1} \left ( f_r, f_{s_{r}}, Z \right )   > \delta_n \left ( \alpha \right ) + \epsilon \right \}.  \end{equation} Then, under the assumption of (\ref{eq:ci}),
\begin{align*}
  P \bigg \{ \phi_{CI} \left ( f_r \right ) = 1 \;\cap\; &r \in \Theta^{MSS}_{\epsilon, n} \bigg \} = P \left \{   \sum_{Z \in \ndata_0} \nabla_{n,1} \left ( f_r, f_{s_{r}}, Z \right )   > \delta_n \left ( \alpha \right ) + \epsilon \cap r \in \Theta^{MSS}_{\epsilon, n} \right \}  \\
 &\overset{(a)}{\le}  P \left \{   \sum_{Z \in \ndata_0} \nabla_{n,1} \left ( f_r, f_{s_{r}}, Z \right )  - \left [ R_n \left (f_r \right ) - R_n \left (f_{s_{r}} \right ) \right ]   > \delta_n \left ( \alpha \right ) \cap r \in \Theta^{MSS}_{\epsilon, n} \right \} \\
  &\le P \left \{   \sum_{Z \in \ndata_0} \nabla_{n,1} \left ( f_r, f_{s_{r}}, Z \right )  - \left [ R_n \left (f_r \right ) - R_n \left (f_{s_{r}} \right ) \right ]   > \delta_n \left ( \alpha \right ) \right \} \\
  &\le \alpha
\end{align*}
where \( (a) \) follows from \( r \in \Theta^{MSS}_{\epsilon, n} \) which implies that \( R_n \left (f_r \right ) - R_n \left ( f_{s_r} \right ) - \epsilon \le 0 \).

There are many methods for deriving concentration inequalities in the form of (\ref{eq:ci}) \citep{boucheron_concentration_2013}. A recent work on MSS showcases a valid construction of (\ref{eq:ci}) using a one-sided empirical Bernstein inequality \citep{takatsu_bridging_2025}. When \( \ndata_0  \) and \( \ndata_1 \) form a partition of \( \ndata\), popular alternative approaches to creating bounds in the form (\ref{eq:ci}) often utilize exponential moment inequalities. For instance, one could assume that \( \nabla_{n,1} \left ( f_r, f_{s_{r}}, Z_0 \right ) - \left [ R_n \left (f_r \right ) - R_n \left ( f_{s_r} \right ) \right ] \) is sub-gaussian, sub-exponential, or sub-weibull. Notably, the sub-gaussian assumption holds whenever \(\nabla_{n,1} \left ( f_r, f_{s_{r}}, Z_0 \right ) \) is almost surely bounded.
 
PAC bounds are a popular concept in statistical learning theory that provide probabilistic guarantees on a learning algorithm's performance. The PAC framework also allows for the computation of data-dependent, uniform bounds of the form of \eqref{eq:ci} \citep{valiant_theory_1984}. In general, if the model class \( \mc_{r}' \subseteq \mc \) where \( f_r, f_{s_r} \in \ac \left ( \mc_{r}' \right ) \) is limited in complexity (as measured by e.g. its cardinality, rademacher complexity, fat shattering dimension, covering number, or VC dimension) and the loss function is light-tailed, then a bound may be derived through an appropriate choice of concentration inequalities \citep{vapnik_nature_2000,mohri_foundations_2018,fisher_all_2019}. Typical assumptions to ensure that the loss is light-tailed are to either assume it is bounded or satisfies an exponential moment inequality \citep{chafai_interactions_2012,boucheron_concentration_2013}. In cases where \( \mc_{r}' \) is very complex (which is typical of ML algorithms), the limiting assumption on the complexity of \( \mc_{r}' \) may be replaced by stability assumptions on \( f_r \) and \( f_{s_r} \) \citep{elisseeff_stability_2005,shalev-shwartz_learnability_2010,feldman_generalization_2018}. Similar to the assumptions of Theorem \ref{thm:sui-cv}, these stability assumptions typically force the difference between \( \ell \left ( f_r \left ( \ndata \right ), Z_0 \right ) - \ell \left ( f_{s_r} \left ( \ndata \right ), Z_0 \right ) \) and \( \ell \left ( f_r \left ( \ndatai \right ), Z_0 \right ) - \ell \left ( f_{s_r} \left ( \ndatai \right ), Z_0 \right ) \) to decrease towards \( 0 \) in probability at a rate dependent on \( n \). Recently, non-vacuous bounds for complex machine learning models such as neural networks have been achieved through PAC bayesian approaches that require suitable choices of data dependent priors \citep{dziugaite_computing_2017}. Current work on PAC bounds focus on developing tighter bounds with milder assumptions. Although noteworthy, PAC bounds are typically not very tight and yield conservative tests.

Another line of work has shown that if the loss is bounded and \( \delta_n \left ( \alpha \right ) \) is chosen dependent on the complexity of \( \mc_{r}' \) where \( f_r, f_{s_{r}} \in \mathcal{F} \left ( \mc_{r}' \right ) \), then a bound of the form of (\ref{eq:ci}) can be achieved \citep[Lemma 23]{fisher_all_2019}. Follow up works have constructed procedures that would allow for computing \( \inf_{t \in T_r} \phi_{CI} \left ( h_{t} \right ) \) in such cases. This would permit an application of Proposition \ref{prop:uniform} to construct a valid MCS test. These methods are viable when \( \mc_{r}' \) is composed of simpler model types such as (regularized) linear models, linear models in a reproducing kernel Hilbert space, rule lists, sparse decision trees, kernel ridge regression models, random forests with a fixed set of pre-trained trees, and sparse generalized additive models \citep{fisher_all_2019, hara_approximate_2018, mata_computing_2022, xin_exploring_2022,laberge_partial_2023,zhong_exploring_2023}.

\subsection{Uniform Coverage}
\label{sec:uniform-coverage}
  Thus far, we have only discussed test construction strategies that yield a confidence set \( \hat{\Theta} \) with point-wise coverage guarantees (\ref{eq:pointwise-mcs-coverage}).  We wish to point out that attaining stronger uniform coverage guarantees of the form
 \begin{equation} \label{eq:uniform-mcs-coverage} \left ( \overset{(as)}{\liminf_{n \to \infty}} \right )  \inf_{P \in \mathcal{P}}  P \left (  \Theta \subseteq \hat{\Theta} \right )  \ge 1-\alpha\end{equation}
 would require controlling the family-wise error rate of the tests \( \psi_r \) for \( r \in [d] \). The simplest approach would be to utilize methods from the multiple comparisons literature. For example, setting the significance level to \( \alpha/d \) would appropriately result in uniform coverage by the union bound if (\ref{eq:test-condition}) is satisfied. As another option, the Benjamini-Hochberg method is a multiple comparisons procedure that fixes some  potential issues of the union bound. Notably, it has already been applied in MSS settings \citep{dey_multiple_2026}.

Alternatively, it has been shown that dimension agnostic uniform coverage is possible in some settings by including a data-dependent screening step \citep{kim_locally_2025}. Such methodology would require constructing an MCS confidence set from a subset of the data, then using the size of the confidence set to estimate an appropriate signifiance level at which to run the MCS procedure on the full sample. These methods, however, are beyond the scope of this paper, and we leave further investigation of the potential of these approaches to future work.

\newpage
\subsection{Box Plots for Real Data Experiments}
\begin{figure}[h!]
  \includegraphics[width=.95\textwidth]{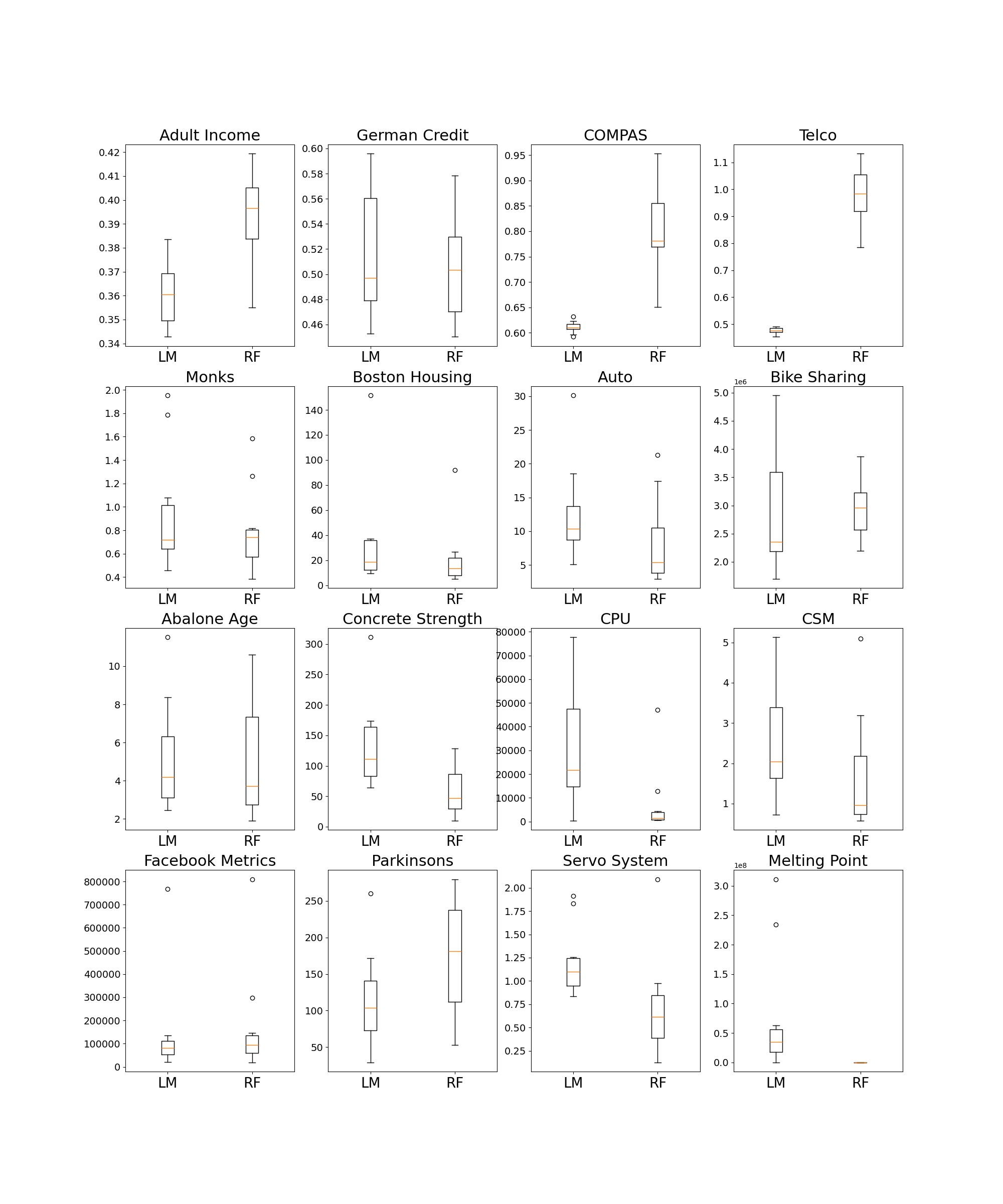}\caption{Box plots representing the 10-fold cross-validation errors for \( f_1 \) and \( f_2 \) that are used in the computation of the test statistic \( \bar{R}_n \left ( f_1, f_{s_1} \right ) \) when using the test \( \phi^{CV}_{CLT} \). LM and RF correspond to classes \( \mc_1 \) and \( \mc_2 \), respectively.}
  \label{fig:rl-cv-results}
\end{figure}


\end{document}